\documentclass[journal, twoside, web]{IEEEtran}
\usepackage[usenames,dvipsnames]{xcolor}
\usepackage{cancel}
\usepackage{amsfonts,amsmath,amssymb}
\usepackage[english]{babel}
\usepackage[normalem]{ulem}
\usepackage{hyperref}

\usepackage{amsthm}
\usepackage{graphicx}
\usepackage{subcaption,caption}
\usepackage{epsfig}
\usepackage{ulem}
\usepackage{tikz}
\usepackage{mathrsfs}
\usepackage{bbold}             
\usepackage{bm}
\usetikzlibrary{shapes,arrows}
\usetikzlibrary{matrix}
\usetikzlibrary{calc,patterns,angles,quotes,arrows,automata}
\graphicspath{ {./imgs/}}

\usepackage[linesnumbered,noend]{algorithm2e}
\RestyleAlgo{ruled}
\SetKwInOut{Input}{Input}
\SetKwInOut{Output}{Output}
\SetKwInOut{Parameters}{Parameters}
\SetKwIF{If}{ElseIf}{Else}{if}{:}{else if}{else}{endif}

\theoremstyle{definition}
\newtheorem{example}{Example}

\theoremstyle{theorem}
\newtheorem{problem}{Problem}
\newtheorem{lemma}{Lemma}
\newtheorem{proposition}{Proposition}

\newtheorem{theorem}{Theorem}
\newtheorem{corollary}{Corollary}
\newtheorem{definition}{Definition}

\theoremstyle{definition}

\theoremstyle{remark}
\newtheorem*{remark}{Remark}

\usepackage[backend=bibtex, style=ieee, sorting=none]{biblatex}
\newcommand{\Ac}{\mathcal{A}}

\newcommand{\Cc}{\mathcal{C}}
\newcommand{\Dc}{\mathcal{D}}
\newcommand{\Ec}{\mathcal{E}}

\newcommand{\Hc}{\mathcal{H}}
\newcommand{\Ic}{\mathcal{I}}
\newcommand{\Jc}{\mathcal{J}}
\newcommand{\Kc}{\mathcal{K}}

\newcommand{\Mc}{\mathcal{M}}

\newcommand{\Qc}{\mathcal{Q}}
\newcommand{\Rc}{\mathcal{R}}
\newcommand{\Sc}{\mathcal{S}}

\newcommand{\As}{\mathscr{A}}
\newcommand{\Bs}{\mathscr{B}}

\newcommand{\Ds}{\mathscr{D}}
\newcommand{\Es}{\mathscr{EL}}

\newcommand{\Ns}{\mathscr{N}}
\newcommand{\Os}{\mathscr{O}}

\newcommand{\Rs}{\mathscr{R}}
\newcommand{\Ss}{\mathscr{S}}

\newcommand{\Vs}{\mathscr{V}}
\newcommand{\Ws}{\mathscr{W}}

\newcommand{\Ys}{\mathscr{Y}}

\newcommand{\Cb}{\mathbb{C}}

\newcommand{\Eb}{\mathbb{E}}

\newcommand{\Nb}{\mathbb{N}}

\newcommand{\Rb}{\mathbb{R}}

\newcommand{\Bf}{\mathfrak{B}}

\newcommand{\Df}{\mathfrak{D}}

\newcommand{\Hf}{\mathfrak{H}}

\newcommand{\Sf}{\mathfrak{S}}

\newcommand{\CE}{\mathbb{E}|}

\newcommand{\SE}{\mathbb{J}|}

\newcommand{\ro}{\mathcal{R}_0}
\newcommand{\eo}{\mathcal{J}_0}
\newcommand{\rs}{\mathcal{R}}
\newcommand{\es}{\mathcal{J}}

\newcommand{\um}{\frac{1}{2}}

\newcommand{\zero}{\mathbb{0}}

\newcommand{\im}{{\rm Im}}

\newcommand{\fix}{{\rm fix}}

\newcommand{\tr}{{\rm tr}}
\newcommand{\diag}{{\rm diag}}
\newcommand{\Span}{{\rm span}}

\newcommand{\supp}{{\rm supp}}

\newcommand{\alg}{{\rm alg}}

\newcommand{\comp}{{\sim}}
\newcommand{\zentrum}{\mathcal{Z}}

\newcommand{\ket}[1]{\left| #1 \right>}
\newcommand{\bra}[1]{\left< #1 \right|}

\newcommand{\inner}[2]{\left< #1,#2 \right>}

\newcommand{\braket}[2]{\left< #1 \middle\vert #2 \right>}
\newcommand{\ketbra}[2]{\left\vert #1 \middle>\middle< #2 \right\vert}

\usepackage{mathabx}
\renewcommand{\check}{\widecheck}
\renewcommand{\bar}{\overline}
\renewcommand{\Tilde}{\widetilde}
\renewcommand{\Hat}{\widehat}

\title{\LARGE \bf 
Model Reduction for Quantum Systems:\\ Discrete-time Quantum Walks and Open Markov Dynamics
}

\author{Tommaso Grigoletto and Francesco Ticozzi
\thanks{T. Grigoletto and F. Ticozzi are with the Department of Information Engineering, University of Padova, Via Gradenigo 6, 35131 Padova, Italy. Emails: 
\href{mailto:tommaso.grigoletto@unipd.it}{\texttt{tommaso.grigoletto@unipd.it}},
\href{mailto:ticozzi@dei.unipd.it}{\texttt{ticozzi@dei.unipd.it}}. F.T. acknowledges funding from the European Union - NextGenerationEU, within the National Center for HPC, Big Data and Quantum Computing (Project No. CN00000013, CN 1,Spoke 10).}}

\begin{document}
\maketitle
\begin{abstract}
     A general approach to obtain reduced models for a wide class of discrete-time quantum systems is proposed. The obtained models not only reproduce exactly the output of a given quantum model, but are also guaranteed to satisfy physical constraints, namely complete positivity and preservation of total probability. A fundamental framework for exact model reduction of quantum systems is constructed leveraging on algebraic methods, as well as novel results on quantum conditional expectations in finite-dimensions. The proposed reduction algorithm is illustrated and tested on prototypical examples, including the quantum walk realizing Grover's algorithm. 
\end{abstract}
\begin{IEEEkeywords}
Model reduction; Quantum systems; Quantum conditional expectations; Quantum walks.
\end{IEEEkeywords}

\section{Introduction}

Finding simpler descriptions for a dynamical model that is too large or complex to study or simulate is arguably a fundamental task in many scientific fields. 
From a dynamical-system viewpoint, the problem has been studied extensively, aiming for both exact reductions, i.e. smaller models that are able to reproduce exactly the target evolution \cite{rosenbrock1970state, kalman1969topics, kalman}, and approximate ones \cite{balanced-truncation}, in the linear as well as in non-linear settings \cite{astolfi_moment_matching}. In many cases, however, one may want to reduce a given model while maintaining certain constraints that characterize the dynamics, e.g. make the model physically admissible or satisfy some locality constraints. This problem, in general, proves difficult to solve: for example, model reduction for linear dynamical systems with positivity constraints is still a partially open issue \cite{farina-benvenuti,benvenuti}.

In this work, we tackle the problem of finding simpler and physically-admissible models that reproduce the output of a wide class of discrete-time quantum dynamical systems. For these systems to be physically admissible, in addition to linearity, minimal requirements are positivity and the preservation of the total probability, i.e. we seek Completely-Positive, Trace-Preserving (CPTP) dynamics. Potential direct applications include alternative, ``compressed'' version of quantum walks \cite{szegedy2004quantum,magniez2007search,kempe2003quantum,portugal2013quantum} as well as reduced models and filters for digital quantum control \cite{Bolognani_Ticozzi, haroche2006exploring}.
Finding the minimal reduced model that preserves the physical constraints also allows for efficient simulations of dynamics on quantum computers \cite{barreiro2011open}. 
The problem of finding computationally-efficient representations is particularly relevant nowadays since the quantum computers available have very limited resources \cite{ma2020quantum}. In addition, such a reduction allows us to understand what are the essential degrees of freedom of a model and which instead can be removed. Moreover, one can use the reduced model to investigate the ``quantumness'' of the variables that cannot be discarded: by studying their algebraic structure one can understand if the system is intrinsically quantum, or it might be represented classically, or even with a hybrid quantum-classical model \cite{diosi2023hybrid, barchielli2023markovian, barchielli2024hybrid}. 

The problem of finding smaller physical representations of quantum systems is of course not new: available model reduction methods such as derivation of master equations, rotating-wave approximations, adiabatic elimination \cite{alicki-lendi,petruccione-open,adiabadic-elimination}, provide approximate models with a trade-off between the accuracy and the simplicity of the reduced model. 
We here focus instead on developing a framework for finding models that are capable of reproducing the output of a given model \textit{exactly}. While this request may seem too stringent, it leads to results and techniques that can be  adapted to approximate model reductions \cite{algebraic-adiabatic}. 

The models of interest in this work are discrete-time, time-invariant CPTP dynamics, paired with an output equation, which we call quantum discrete-time semigroup with output (QSO) models, as defined in Sec. \ref{sec:model}. As outputs of interest we consider linear functions of the state: these allow us to cover single-time (unconditional) probability of given events, expectations of  observables of interest, correlation functions, as well as reduced/marginal states. In particular, quantum walks and related algorithms can be recast as iterations of quantum CPTP dynamics, for which the output of interest is the probability distribution of a given observable \cite{venegas2012quantum, magniez2007search}.
In this paper, where substantial work is devoted to building the foundation for a systematic model-reduction framework, we do not explicitly include the evolution of states conditioned on previous measurements, i.e. quantum trajectories and quantum filters \cite{bouten2007introduction, benoist2023limit, monras2010hidden, Schonhuth2011}. 
The seemingly simpler problem of reproducing the single-time probabilities treated here is already challenging on its own, and has key applications of interest, including the analysis, reduction and simulation of quantum walks and quantum algorithms. In fact, this work has been partially motivated by the results of \cite{apers2018simulation, apers2021characterizing}, where the simulatability of quantum walks with classical resources has been used to successfully replicate their mixing speedup and understanding its origin. 
The results we develop in this work can be extended to this scenario adapting the approach of \cite{tac2023, ito1992identifiability} from classical Markov processes to non-commutative ones, as we detail in \cite{letters}, covering the class quantum hidden Markov models as introduced in \cite{monras2010hidden}. 
With respect to classical hidden Markov models, the quantum framework introduces some interesting peculiarities, showing that the minimal model for conditional dynamics could be either larger or smaller than the unconditional one derived here \cite{letters}.  Other extensions of the methods of this paper to continuous-time dynamics, with or without measurements, are discussed in the follow-up works \cite{grigoletto2025exactmodelreductioncontinuoustime, grigoletto2025quantummodelreductioncontinuoustime}. 

Note that the problem of model reduction is connected yet distinct from the (completely) positive realization problem \cite{farina-benvenuti,monras2016quantum}. In particular, the realization problem assumes to have access to some measured quantity of a system at hand from which one desires to reconstruct a model that explains the observed data while, in the model reduction problem, we start from a given physical model and we aim to find a smaller representation that is still physical. Combining both tasks, one can aim to construct a small, possibly minimal, positive realization that explains the observed data.

The method we are proposing to tackle exact model reduction for quantum systems hinges on key tools from linear system theory, namely reachability and observability analysis  and  realization theory \cite{kalman1969topics, rosenbrock1970state, kalman}. Our quantum reachability and observability analysis  is inspired by the one introduced for continuous-time, unitarily evolving controlled quantum systems \cite{dalessandro_quantum_2003}, extended to open system evolution.  Recently, it has been shown \cite{cdc2022} how these tools can be used to find a minimal linear reduction for discrete-time quantum systems, albeit the latter need not satisfy CPTP constraints. The system-theoretic ideas and these preliminary results are recalled in Sec. \ref{sec:linear_model_reduction}, after the notation, models, and problem are presented in Sec. \ref{sec:model}.

Note that \cite{kumar2014model} takes a similar approach to model reduction, but only for unitary evolution and pure initial conditions, while the output quantities of interest are not exploited. In this simpler setting, the fact that the reduced model retains its CPTP property is always guaranteed, as one is projecting an Hamiltonian onto an invariant subspace. As we shall see, when quantum channels and mixed states are considered ensuring this property is more involved.
To the authors' knowledge, no other systematic methods for exact CP model reduction for quantum systems are available.

The core idea of this work is to enlarge the minimal linear reduction subspaces to a c $*$-algebras, and thus to a quantum probability space \cite{maassen2010quantum, beny2015algebraic, kr1992introduction}. The reduction is then obtained using CPTP projections. Our construction exploits the theory of (finite-dimensional) quantum conditional expectations and their duals (CPTP projections), that we develop with some new results (see \cite{takesaki1979theory, blackadar2006operator,petz2007quantum}, and our essential review in Sec. \ref{sec:Conditional_expectations}). In particular, in Sec. \ref{sec:minimal_distorted_algebras}, we show how to find the minimal set that includes a given subspace and admits a CPTP projection. Remarkably, this (sub)problem is connected, and in certain cases equivalent, to the problem of finding the minimal sufficient algebra that allows to discern a parametric family of quantum states \cite{jenvcova2006sufficiency}, or to find the minimal class of states, containing certain ones, that can be left invariant by a CPTP map \cite{PhysRevA.66.022318}.

The main reduction algorithms are presented in Sec. \ref{sec:reduction_algorithms}.
Since classical Markov chains can be viewed as a discrete-time quantum systems restricted to an invariant  abelian algebra, the results presented here also cover and improve the existing one for the classical case. Compared to those proposed in \cite{tac2023},  for {\em classical Hidden Markov models, the results presented here differ in two significant ways. First, for non-commutative algebras, even the existence of CPTP projections is not guaranteed, and hence the construction we propose requires the foundational results presented in Sec. \ref{sec:minimal_distorted_algebras}. Second, we here consider an iterative model reduction algorithm. This allows us to overcome the technical difficulties of the non-iterative algorithm proposed in \cite{tac2023}, which are recalled in Sec. \ref{sec:reduction_algorithms}, and avoid the need for additional conditions in order to have effective reductions. 
Examples of application are discussed in Sec. \ref{sec:examples}, including the well-known quantum-walk implementation of  Grover's algorithm \cite{grover}.

\section{Models and problem definition}\label{sec:model}
\subsection{Notation}
In this work, we denote by $\mathcal{H}$ the (finite-dimensional Hilbert space $\mathcal{H}\simeq\mathbb{C}^n$ and by $\mathfrak{B}(\mathcal{H})\simeq\mathbb{C}^{n\times n}$ the set of bounded operators in $\mathcal{H}$. We then use capital letters to denote operators, $X\in\mathfrak{B}(\mathcal{H})$, and the $\dag$-superscript to denote their adjoints (as well as the transpose and conjugate of their matrix representation). $\mathfrak{H}(\mathcal{H})$ represents the set of self-adjoint operators, i.e. $\mathfrak{H}(\mathcal{H}) = \{X\in\mathfrak{B}(\mathcal{H})|X = X^\dag\}$. $\mathfrak{D}(\mathcal{H})$  denote  the set of density operators, i.e. $\mathfrak{D}(\mathcal{H})=\{X\in\mathfrak{B}(\mathcal{H})|X=X^\dag\geq0, \, \tr[X]=1\}$. With few exceptions, the ``script'' notation is used to denote operator spaces, e.g. $\mathscr{A}, \mathscr{B}\subseteq\mathfrak{B}(\mathcal{H})$ and the ``calligraphic'' notation to denote super-operators, e.g. $\mathcal{A}:\mathfrak{B}(\mathcal{H})\to\mathfrak{B}(\mathcal{H})$. 

We  denote with $\braket{\cdot}{\cdot}$ the standard inner product over $\Hc$, and with $\inner{\cdot}{\cdot}_{HS}$ the  Hilbert-Schmidt inner product for $\Bf(\Hc)$, that is $\inner{X}{Y}_{HC} = \tr(X^\dag Y)$. In this work we employ also different inner products for $\Bf(\Hc)$: from Friedrichs representation theorem, see e.g. \cite[Theorem 4]{albeverio2003lectures}, every inner product is associated to a self-adjoint positive-definite super operator $\Qc:\Bf(\Hc)\to\Bf(\Hc)$, $\Qc=\Qc^\dag>0$ as $\inner{X}{Y}_\Qc = \inner{X}{\Qc(Y)}_{HS}$. In the rest of this work, we  denote inner products by their positive super operator, i.e. $\inner{\cdot}{\cdot}_{\Qc}$ denotes the inner product defined as $\inner{\cdot}{\Qc(\cdot)}_{HS}.$  Orthogonality with respect to a modified inner product $\inner{\cdot}{\cdot}_\Qc$, is referred $\Qc$-orthogonality.

If $\Vs$ and $\Ws$ are linear operator subspaces, $\Vs +\Ws$ is used for the minimal subspace containing both spaces, while $\Vs \oplus \Ws$ is used to indicate the same while also specifying that the two subspaces are orthogonal with respect to the Hilbert-Schmidt inner product. The symbol $\oplus_\Qc$ is used in case the orthogonality of the considered spaces is with respect to the inner product $\inner{\cdot}{\cdot}_\Qc$. 

As usual, $[\cdot,\cdot]$ denotes the commutator of two operators. Given an operator $X\in\Bf(\Hc)$, we define the support of $X$ as the subspace orthogonal to its kernel, i.e. $\supp(X)\equiv (\ker(X))^\perp$.
Similarly, given a set of operators $\Sc\subseteq\Bf(\Hc)$ we define its support as the sum of the supports of the operators it contains, i.e $\supp(\Sc) \equiv \sum_{X\in\Sc}\supp(X).$ 

In this work, by $*$-algebra we mean a linear operator space closed under composition and the adjoint involution $\dag$. A $*$-algebra $\As$ is said to be unital if it contains the identity, i.e. $I\in\As$. Given a set of operators $\Ss$ we denote by $\Ss'$ its commutant, i.e. $\Ss' \equiv \{X\in\Bf(\Hc)| [S,X]=0,\, \forall S\in\Ss\}$ and given a $*$-algebra $\As$ we denote by $\zentrum(\As)$ its center, i.e. $\zentrum(\As) \equiv \As\cap\As'$. For a $*$-algebra $\As$, the symbol $\dim(\As)$ denotes the dimension of $\As$ intended as a linear operator space, i.e. the number of linearly independent operators that generate $\As$ through linear combination. }
With some abuse of notation we will denote with $\Hf(\As)$ and $\Df(\As)$ the set of Hermitian operators and density operators contained in the algebra $\As$ respectively, i.e. $\Hf(\As) \equiv \Hf(\Hc)\cap\As$ and $\Df(\As) \equiv \Df(\Hc)\cap\As$. 
Further relevant facts about $*$-algebras, their representations  and their properties will be introduced when needed, in particular in Sec. \ref{sec:Conditional_expectations} and \ref{sec:minimal_distorted_algebras}.

For any positivedefinite operator $\sigma$ we define the symmetry-preserving multiplicative operator $\Dc_{\sigma}(X)\equiv\sigma^\um X \sigma^\um.$

\subsection{Quantum discrete-time semigroups with outputs}

In this work we describe quantum systems in an algebraic framework, following e.g. \cite{maassen2010quantum,beny2015algebraic,bouten2007introduction,kr1992introduction}.  The quantum system of interest is associated with a unital $*$-algebra $\Bf(\Hc)$. The self-adjoint operators in $\Hf(\Hc)$ represent the observable quantities, while the available information on the system is condensed in a linear, positive and normalized functional $\Eb_{\rho}[\cdot]= \left<\rho,\cdot\right>_{HS}$, where $\rho\in\Df(\Hc)$ is the {\em state} of the system, providing the expectation values of the observables with respect to a given state. We next introduce the class of dynamics of interest, together with the key assumptions we shall leverage to reduce the model: the interest in reproducing {\em only} a subset of time-varying quantities of interest, and the knowledge that the initial state for the dynamics belongs to a {\em limited subset of initial conditions}.

\subsubsection{Dynamics}
General, physically admissible state transformations are associated to completely-positive (CP), trace-preserving (TP) linear maps  $\mathcal{A}:\Bf(\Hc)\mapsto\Bf(\Hc)$ \cite{nielsen_chuang_2010}.
Any CP map admits an operator-sum representation (also known as Kraus representation \cite{nielsen_chuang_2010}): $\mathcal{A}(\rho)=\sum_k A_k\rho A_k^\dag,$ which is TP if and only if  $\sum_k A_k^\dag A_k=I.$ The last property is also equivalent to the (Hilbert-Schmidt) adjoint map $\Ac^\dag$ being unital: $\mathcal{A}^\dag(I)=I.$ 
CPTP maps are the non-commutative equivalent of stochastic maps between finite probability spaces.
In this work, we assume the dynamic is in discrete-time $(t\in\mathbb{N}),$ time-homogeneous and  Markovian, and hence associated to iterations of a given CPTP  map: \[\rho(t+1)={\cal A}[\rho(t)].\] 
Alternatively, and equivalently in terms of the evolution of the expectation of observables, dynamics can be described as the action of the CP and unital dual map onto the observable of interest, i.e. $C(t+1) = \Ac^\dag[C(t)],$ while the state is time invariant. As standard in the physics literature, we refer to the first scenario as Schr\"{o}dinger picture and to the latter as Heisenberg picture.

\subsubsection{Quantities of interest and linear output maps}

In many practical cases, we are not actually interested in the entire information contained in the system's state $\rho$, but only in studying (or simulating) the evolution of  a few linear function of the system's state. Typical linear functions include:
\begin{itemize}
    \item the {\em probabilities} associated to an observable $C\in\Hf(\Hc)$ with $C=\sum_j c_j \pi_j$, $p_j = \inner{\pi_j}{\rho}_{HS}$;
    \item the {\em expectation value of multiple observable of interest}, say $\{C_j\}$, $\Eb_\rho[C_j] = \inner{C_j}{\rho}_{HS}$;
    \item the {\em reduced state} of a bipartite quantum system, say $\tr_B[\rho]$ whenever $\Hc = \Hc_A\otimes\Hc_B$.
    \item the {\em cross-correlation} of two observables $C_1,C_2\in\Bs$, as $\Eb_{\rho}[C_1C_2].$ Notice that in general $C_1C_2$ needs not be Hermitian.
\end{itemize}

To capture all of this cases in a single mathematical object we define a vector space $\Ys,$ called an \textit{output space}, and an \textit{output map} $\Cc:\Bf(\Hc)\to\Ys,$ which we assume to be linear and refer to the time-dependent function $\Cc(\rho(t))$ as the \textit{output}. 

For example, in the first three scenarios above, we can consider $\Cc[\cdot] = \sum_j e_j \inner{\pi_j}{\cdot}_{HS},$ $\Cc[\cdot] = \sum_j e_j \inner{C_j}{\cdot}_{HS},$ and $\Cc = \tr_B,$ respectively.

Viceversa, any linear output map can be equivalently represented as a family of expectations $\Eb_\rho[C_j]$ for a finite set of (non necessarily hermitian) operators $\{C_j\}\in\Bf(\Hc)$, if one adopts a vector representation of the output space. We can then consider without loss of generality $\Ys\equiv\Rb^m,$ and $\Cc(\cdot) = \sum_j e_j \inner{C_j}{\cdot}_{HS}$ where $\{e_j\}$ is an orthogonal base for $\Ys$ and with $\{C_j\}$ a set of operators. It is thus equivalent to assume to be provided with an output map $\Cc$ or a finite set of operators of interest $\{C_j\}\subset \mathfrak{B}({\Hc})$

\subsubsection{Initial states}

As it is often the case in experiments and quantum algorithms, the initial conditions of interest are often restricted to a (small) subset of density operators, sometimes limited to a single pure or thermal state. For this reason, we assume to be given a {\em finite set $\Sf\subseteq\Df(\Hc)$} that contains all the initial conditions of interest. By linearity of the evolution, we can relax this assumption to any intersection of a (finite-dimensional) operator affine space and the set of density operators.

\subsubsection{Equivalent dynamical models}

Under the assumptions introduced so far, the  quantities of interest in our model are represented by the {\em output}, whose components are the time-dependent expectations \[{\cal C}_j(\rho_k)=\Eb_{\rho_k}[\Ac^{\dag t}(C_j)]\] for all $t\in\Nb$, and for each observable of interest $C_j$ and each initial state $\rho_k\in\Sf$. In this work we aim to find another Hilbert space $\check{\Hc}$, another set of observables $\check{C}_j\in\Hf(\check{\Hc})$, another set of states $\check{\rho}_k\in\Df(\check{\Hc})$, and another CPTP evolution map $\check{\Ac}:\Df(\check{\Hc})\to\Df(\check{\Hc})$ such that
\[\Eb_{\rho_k}[\Ac^{\dag t}(C_j)] = \Eb_{\check{\rho}_k}[\check{\Ac}^{\dag t}(\check{C}_j)]\quad \forall j,\, \forall k,\, \forall t\geq0. \]
Whenever this condition is satisfied we say that these two models are \textit{equivalent} with respect to the output. 

The same property can be expressed in Schr\"{o}dinger picture: We say that the two models are equivalent whenever
\(\inner{C_j}{\Ac^t[\rho_k]}_{HS} = \big<\check{C_j},\check{\Ac}^t[\check{\rho}_k]\big>_{HS}\), $\forall j, k$ and $ \forall t\geq0,$ or more compactly:
\[\Cc \Ac^t[\rho_k] = \check{\Cc} \check{\Ac}^t[\check{\rho}_k]\quad, \forall k, \forall t\geq0\]
where $\check{\Cc}$ is the linear output map induced by the operators $\check{C}_j$.
Albeit the Heisenberg picture is more natural in an algebraic setting, in the following we mainly use the Schr\"{o}dinger picture. The reason is twofold: in the applications of interest for our reduction methods (open quantum systems, quantum control and quantum information processing), considering state evolution is standard, and the derivation of the reduction methods leveraging on the knowledge of the initial conditions, which are arguably more involved, becomes more direct in this picture. All of the results we derive could, however, be equivalently derived in Heisenberg picture.

\begin{remark} Note that, albeit we defined the equivalent model with operators in $\Bf(\check{\Hc}),$ it is possible, in principle, that the the whole algebra $\Bf(\check{\Hc})$ is not needed to describe the information captured by the model. Specifically, it is possible that the evolved operators of interest $\check{\Ac}^{\dag t}[\check{C}_j]$ always belong to a unital {\em $*$-subalgebra} $\check{\As}\subseteq\Bf(\check{\Hc})$.  
In this case, we can show that there is always a $\check{\rho}\in\check{\As}$ which yields and equivalent model. 

\begin{proposition}
    \label{prop:state_inside_algebra}
    Consider a unital $*$-subalgebra $\Bs\subset\Bf(\Hc)$ and a state $\rho\in\Df(\Hc)$. Then there exist a density operator $\rho_\Bs\in\Df(\Bs)$ such that, for all $B\in\Bs$,
    \[\inner{\rho}{B}_{HS} = \inner{\rho_\Bs}{B}_{HS}.\]
\end{proposition}
\begin{proof}
    Let $\Pi:\Bf(\Hc)\to\Bf(\Hc)$, be the orthogonal projector onto $\Bs$ with respect to $\inner{\cdot}{\cdot}_{HS}$, i.e.  ${\rm Im}(\Pi)=\Bs$ and such that $\Pi=\Pi^\dag=\Pi^2$.
    Since $\Bs$ is unital then $\Pi$ is unital and, by Tomiyama Theorem \cite{tomiyama1957projection}, is also CP. As consequence of the fact that $\Pi$ is self-adjoint $\Pi$ is also CPTP. Then \(\inner{\rho}{B}_{HS} = \inner{\rho}{\Pi(B)}_{HS} = \inner{\Pi(\rho)}{B}_{HS}\) where $\Pi(\rho)\in\Bs$ is a density operator since $\Pi$ is CPTP. 
\end{proof}

For this reason, we can then assume that $\check{\Ac}:\check{\As}\to\check{\As}$ thus restricting the entire new model onto the subalgebra $\check{\As}$. 
Ideally, $\check{\As}$ would be the smallest algebra that supports the probability space necessary to describe the expectations of interest. 
\end{remark}

Putting together all the elements that we introduced so far, we can formally define the class dynamical models of interest for this work.
\begin{definition}[QSO model] 
\label{def:QSO_model}
A {\bf quantum discrete-time semigroup with output (QSO)} is a model 
of the form:
 \begin{equation}
    \begin{cases}
    \rho(t+1) = \mathcal{A}[\rho(t)]\\
    Y(t) = \mathcal{C}[\rho(t)]
    \end{cases}\quad \rho(0)\in\mathfrak{S}.
    \label{eqn:single_time_model}
\end{equation}
where  $\mathscr{B}\subseteq\Bf(\Hc)$ is a unital $*$-algebra,  $\mathcal{A}[\cdot]:\mathscr{B}\mapsto\mathscr{B}$ a CPTP map,  $\mathcal{C}[\cdot]:\mathscr{B}\mapsto\mathscr{Y}$ a linear output map,  $\mathfrak{S}\subseteq\mathfrak{D}(\mathscr{B})$ and a set of initial conditions.
The model is thus specified by the 5-tuple $(\mathscr{B},\mathscr{Y},\mathcal{A},\mathcal{C}, \mathfrak{S})$.
\end{definition}
Crucially, considering models defined on a supporting (proper) subalgebra, rather than a full operator space $\Bf(\Hc),$ allows us to seek further reductions on an already partially reduced model.
For the sake of simplicity, when clear from the context, we specify a QSO model using only the triplet $(\mathcal{A},\mathcal{C}, \mathfrak{S}),$ leaving the state space and the output space intended. 
When all the initial conditions or all states are of interest, one can choose $\Sf=\Df(\Hc)$ or $\Cc=\Ic.$ respectively.

\subsection{Model reduction task and motivation}
In this work, we propose a model reduction scheme that leverages two assumptions that are often verified in the scenarios of interest: the fact that we are only interested in reproducing the output of a given dynamics, and the fact that we restrict our attention to a restricted set of initial conditions.
The main problem we address in this work is the following.
\begin{problem}[QSO model reduction]
\label{prob:single_time}
Given a QSO $(\mathscr{B},\mathscr{Y},\mathcal{A},\mathcal{C}, \mathfrak{S})$ find an equivalent QSO $(\check{\mathscr{A}},\mathscr{Y}, \check{\mathcal{A}},\check{\mathcal{C}}, \check{\mathfrak{S}})$, with $\dim(\check{\mathscr{A}})\leq\dim(\mathscr{B})$, and a positive and trace preserving map $\Phi:\mathscr{B}\to\check{\mathscr{A}}$ such that, for every initial condition $\rho_0\in\mathfrak{S}$ and every time $t\geq0$, the two models provide the same output i.e. 
\[\mathcal{C}\mathcal{A}^t[\rho_0] = \check{\mathcal{C}}\check{\mathcal{A}}^t\Phi[\rho_0].\] 
\end{problem}
Note that, in most practical cases, $\Bs$ can be taken to be $\Bs=\Bf(\Hc)$.  Moreover, the procedure we propose always returns a map $\Phi$ that is actually CPTP. 

If one is only interested in reproducing the outputs of a given QSO model and is {\em not} interested in retaining the CP character of the evolution, the model reduction problem is significantly simpler and an optimal solution can be devised, as we shall see in the next sections. Such solution, however, has in general no physical interpretability and is at risk of providing nonphysical (nonpositive) predictions in the presence of small errors. 

The constraint that the reduced evolution map is CPTP represents in fact both the main novelty, as well as  the key challenge, in this work.
A particularly  case of interest that further motivates the need for a CP reduction arises when one aims to implement an efficient simulations on a quantum computer: having a CP dynamics would allow for a directly implementable system.

\subsection{Connections with hidden Markov models}
Allowing for a general linear output function, the class QSO model we introduced naturally covers a number of settings of interest for applications, with algorithms based on quantum walks and models of open system being the most natural ones. In addition, it has a natural connection to, and it has been motivated by, classical hidden Markov models \cite{vidyasagar2011complete,ito1992identifiability}. In that setting, the full systems evolves like a (large) Markov chain $X(t)$, but the available observations consist only on a function $Y(t) = f[X(t)]$ of this unaccessible process, which needs not be a Markov process itself. If the quantity of interest for a hidden Markov model is the marginal distribution of $Y$ at different times \cite{tac2023}, it can be seen as a particular case of a QSO models where the underlying algebra is commutative.
For these classical models, algebraic methods for stochasticity-preserving reductions have been developed in \cite{tac2023}, and their relation with the presented work is discussed in Appendix \ref{sec:hmm}.
However, for hidden Markov model one is typically concerned with reproducing the output full statistics, and not only the one-time marginals. In order to cover this situation in the quantum case, it is necessary to endow our QSO model with the conditioning effects emerging from quantum (generalized) measurements, leading to stochastic quantum trajectories or filtering equations \cite{bouten2007introduction,benoist2023limit, wiseman2009quantum, barchielli2009quantum, carmichael2009open, belavkin1992quantum}. Discrete-time stochastic processes of this kind, where the relevant output is a classical stochastic process emerging from a quantum evolution conditioned on measurements, have in fact been named {\em quantum hidden Markov models} \cite{monras2010hidden,monras2016quantum,cholewa2017quantum}  in analogy with their classical counterpart. The foundational methods introduced here for CPTP reduction can be extended to these cases, essentially modifying the notion of reachable and observable subspaces we shall introduce in Sec. \ref{sec:linear_model_reduction} to include the effect of conditioning. This has been done explicitly, building on the results of this work, in \cite{letters,grigoletto2025quantummodelreductioncontinuoustime}, for models subjected to discrete and continuous measurements, respectively.
Interestingly, while in the classical case a model able to reproduce the full statistics of a process is {\em always} able to also reproduce the unconditional single-time distributions, this is not always the case in the quantum setting, due to the non-commutative nature of the probabilistic structure, as explicitly shown in \cite{letters}.

\section{Reachability, observability and\\ optimal linear reductions}
\label{sec:linear_model_reduction}

In this section, we  collect a series of definitions and results adapted from linear system theory \cite{wonham,kalman,rosenbrock1970state} that provides the basis for solving Problem \ref{prob:single_time} and, at the same time provide the solution to a related and simpler problem. 

\subsection{Reachable subspace}
Given an initial condition $\rho(0)=\rho_0$ a trajectory of the model $(\mathcal{A,C},\mathfrak{S})$ from $\rho_0\in\Sf$ is defined as the set of states $\rho(t)$ that the system assumes at non-negative times $\mathfrak{T}_{\rho_0}\equiv\{\rho(t)=\mathcal{A}^t[\rho_0], t\geq0\}$. Consequently, the set of trajectories starting from points contained in $\mathfrak{S}$ is defined as the union of the single trajectories $\mathfrak{T}_\mathfrak{S} = \bigcup_{\rho_0\in\mathfrak{S}}\mathfrak{T}_{\rho_0}$.

\begin{definition}[Reachable subspace from $\Sf$]
    We call reachable subspace (from $\Sf$) of the model $(\mathcal{A}, \mathcal{C}, \mathfrak{S})$ the linear space generated by the set of trajectories $\mathfrak{T}_\mathfrak{S}$:
    \begin{equation}
        \Rs : = \Span\{\mathcal{A}^t[\rho_0],\,  t\geq0, \, \rho_0\in\mathfrak{S}\}.
        \label{eq:reachable_space}
    \end{equation}
\end{definition}
It can be proven, in analogy of what is done in for classical linear systems with inputs \cite{cdc2022,kalman}, that $\Rs$ is the smallest $\mathcal{A}$-invariant subspace of $\mathscr{A}$ that contains $\Span\{\mathfrak{S}\}$. Using Caley-Hamilton, one can also prove that $\Rs = \Span\{\Ac^t[\rho_0],\, t=0,\dots,n^2-1, \,\rho_0\in\Sf\}$ where $n=\dim(\Hc)$.

\subsection{Non-observable subspace}
\label{sec:non-observable_subspace}
We next characterize the set of operators that produce no output at all times.
\begin{definition}[Non-observable subspace]
    The non-observable subspace of the model $(\mathcal{A}, \mathcal{C}, \mathfrak{S})$ is the subspace
    \begin{equation}
        \mathscr{N}\equiv \{X \in\mathscr{B} | \mathcal{C}\mathcal{A}^t[X]=0, \, \forall t\geq0\}.
        \label{eq:non_observable_subspace}
    \end{equation}
\end{definition}
It is possible to prove that $\mathscr{N}$ is the largest $\mathcal{A}$-invariant subspace contained in $\ker{\mathcal{C}}$ \cite{kalman}.

The non-observable subspace can be used to characterize the set of states that are not distinguishable from the output at any time. If we consider two initial states $\rho_1,\rho_2\in\mathfrak{D}(\Hc)$ and  $\rho_1-\rho_2\in\mathscr{N}$, then their output is equivalent at all times, namely $\mathcal{C}[\mathcal{A}^t[\rho_1]]=\mathcal{C}[\mathcal{A}^t[\rho_2]]$ for all $t\geq0$. 
 
As described before, because we assumed $\Cc$ to be linear, we can assume to have access to the set of operators $\{C_j\}$ that define $\Cc$.
We can then compute, see e.g. \cite{cdc2022}, the space orthogonal to $\Ns$, defined as $\Ns^\perp = \{X\in\Bs|\inner{X}{Y}_{HS}=0, \, \forall Y\in\Ns\},$ in the following way:
\[\mathscr{N}^\perp = \Span\{\mathcal{A}^{\dag t}[C_j], \, t=0,\dots, n^2-1\}. \]

\subsection{Effective subspace and minimal linear realizations}

\begin{definition}[Effective subspace]
    Given the reachable and non-observable subspace associated to model $(\mathcal{A,C}, \mathfrak{S})$, i.e. $\mathscr{R}$ and $\mathscr{N}$ we define an {\em effective subspace} $\mathscr{E}$ as a (non-necessarily orthogonal) complement of the intersection $\mathscr{R}\cap\mathscr{N}$ to $\mathscr{R}$, i.e. $\mathscr E$ is any subspace such that $\mathscr{R} = \mathscr{E}+(\mathscr{R}\cap\mathscr{N})$.
\end{definition}

\begin{remark}
    It should be noted that the choice of $\Es$ is not unique, in fact, any representative of the quotient space $\Rs/(\Rs\cap\Ns)$ is a valid choice for $\Es$. Usually, when solving a model reduction problem in control system theory, every choice of a representative of $\Rs/(\Rs\cap\Ns)$ is equivalent. 
    When solving Problem \ref{prob:single_time} however, we aim to find a CPTP reduced dynamics. The additional positivity constraints in this scenario make the choice of a particular effective subspace relevant in the effort of finding the minimal model.
\end{remark}

Intuitively, the effective subspace contains all the states that can be reached by the model minus the ones that can not be distinguished from the output. This is confirmed by the fact that the projection onto the effective subspace solves a simpler version of Problem \ref{prob:single_time}. 

\begin{problem}[Linear model reduction]
\label{prob:linear_mr}
Given a QSO $(\mathscr{B},\mathscr{Y},\mathcal{A},\mathcal{C}, \mathfrak{S})$ find a linear model $(\mathscr{V},\mathscr{Y}, \Ac_L,\Cc_L, \Sf_L)$, of equations
\[\begin{cases}
    X(t+1) = \Ac_L[X(t)] & X\in\mathscr{V} \\
    Y(t) = \Cc_L[X(t)]& Y\in\mathscr{Y}
\end{cases} \quad X(0) \in{\mathfrak{S}_L}\]
possibly smaller, i.e. $\dim(\Vs)\leq\dim(\mathscr{B})$, and a linear map $\Phi:\mathscr{B}\to\mathscr{V}$, $\Phi[\rho_0] = X(0)$ such that, for every initial condition $\rho_0\in\mathfrak{S}$ and every time $t\geq0$, the two models provide the same output i.e. 
\[\mathcal{C}\mathcal{A}^t[\rho_0] = \Cc_L{\Ac_L}^t\Phi[\rho_0].\] 
\end{problem}

The difference between Problems \ref{prob:single_time} and \ref{prob:linear_mr} is subtle but crucial. Problem \ref{prob:single_time} aims to find a reduced QSO model, namely the map $\check{\Ac}$ has to be CPTP and the state space $\check{\As}$ needs to be a $*$-algebra. In Problem \ref{prob:linear_mr}, this assumption is relaxed, $\Ac_L$ only needs to be linear, and $\Vs$ is simply an operator space, thus simplifying the problem.

An equivalent version of this problem (with zero initial condition but with linear inputs) has been solved originally by Rosenbrock \cite{rosenbrock1970state}: it has been proven recently  that the same approach can be used to solve Problem \ref{prob:linear_mr} \cite{cdc2022,tac2023}. We report here the result for completeness.

\begin{proposition}
Let $\mathcal{E}$ be an effective subspace of the QSO model $(\mathscr{B},\mathcal{Y,A,C,S})$ and $\Pi_\mathcal{E}:\mathscr{A}\to\mathcal{E}$ be the orthogonal (with respect to $\inner{\cdot}{\cdot}_{HS}$) projection onto $\mathcal{E}$. Then $\mathcal{E}$ is a subspace of minimal dimension such that, defining $\Ac_L = \Pi_\mathcal{E}\mathcal{A}\Pi_\mathcal{E}$, $\Cc_L = \mathcal{C}\Pi_\mathcal{E}$, \[\mathcal{C}\mathcal{A}^t[\rho_0] = {\Cc_L}{\Ac_L}^t\Pi_\mathcal{E}[\rho_0].\]
\end{proposition}

The proof of this Proposition can be found in \cite{cdc2022}. While this result provides the best linear exact reductions for a system, these do not preserve two key aspects of a QSO model: (i) the state space should be associated with an operator algebra; (ii) the dynamics should be CPTP. The rest of the paper is devoted to constructing reductions that are indeed QSO. Necessarily, these models  have a dimension greater or equal to the ones obtained via reduction to an effective subspace.

\section{Conditional expectations, state extensions and their CPTP factorizations}
\label{sec:Conditional_expectations}

In this and the next section we assume, for the sake of simplicity, that $\Bs=\frak{B}(\Hc)\simeq\mathbb{C}^{n\times n}$. All the material we present can be appropriately adapted to the case where $\Bs\subsetneq\Bf(\Hc)$.

\subsection{Conditional Expectations and State Extensions: A Brief Review}
Let $\As$ be a $*$-subalgebra of $\Bs$, i.e. $\As\subseteq\Bs$. It is known that, via some unitary $U,$ any $*$-subalgebra, as well as its commutant and center, can be put in an orthogonal block decomposition (also called Wedderburn decomposition) \cite{wedderburn1908hypercomplex}, associated to a decomposition of the underlying vector space as:
\begin{equation}\label{eq:hildec}
    \Hc 
    =\bigoplus_\ell \Hc_{S,\ell}\otimes\Hc_{F,\ell}\oplus\Hc_{R},
\end{equation}
so that:

\begin{align}
    \As&=U\left(\bigoplus_\ell \Bf(\Hc_{S,\ell})\otimes I_{F,\ell}\oplus \zero_{R}\right)U^\dag,
    \label{eq:algblock}\\
    \As'&= U\left(\bigoplus_\ell I_{S,\ell}\otimes \Bf(\Hc_{F,\ell})\oplus I_{R}\right)U^\dag,
    \label{eq:commblock}\\
    \zentrum(\As)&= U\left(\bigoplus_\ell \Cb_\ell I_{S,\ell}\otimes I_{F,\ell}\oplus \zero_{R}\right)U^\dag,
    \label{eq:zenblock}
\end{align}
In the following we focus on unital $*$-algebras, corresponding to $\Hc_R=0.$

A {\em conditional expectation} $\CE_{\As}:\Bs\to\Bs$, is a CP, unital map with ${\rm Im}(\CE_\As) = \As$ that satisfies \cite{petz2007quantum}:
\begin{equation}
    \CE_{\As}(AB)=A\CE_{\As}(B)
\end{equation}
for every $A\in\As$, $B\in\Bs$.
By choosing $B=I$ we have that $\CE_{\As}$ acts identically on $\As,$ so ${\rm Fix}(\CE_\As) = {\rm Im}(\CE_\As)=\As,$ and $\CE_\As^2=\CE_\As$, thus $\CE_\As$ is a projector onto $\As$.  However, it is not in general self-adjoint, and thus not an orthogonal projection, with respect to the standard Hilbert-Schmidt product. 

A {\em state extension}, $\SE_\As:\Bs\to\Bs$, is the dual of a conditional expectation with respect to the Hilbert-Schmidt inner product \cite[Chapter 9]{petz2007quantum}, namely the unique CPTP map $\SE_{\As}\equiv\CE_\As^\dag$, satisfying:
\[ \inner{\SE_{\As}(\rho)}{X}_{HS}=\inner{\rho}{\CE_{\As}(X)}_{HS},\quad \forall \rho, X\in\Bf(\Hc).\]
Since $\CE_\As$ is idempotent, then also $\SE_{\As}^2=\SE_{\As}.$ Moreover, if  for some $\bar\rho\in\Df(\Hc),$ we have $\SE_{\As}(\bar \rho)=\bar\rho$ then $\CE_{\As}$ is said to preserve $\bar\rho.$ In the following we say that $\SE_\As$ is a state extension \textit{associated to $\As$} if $\SE_{\As}^\dag=\CE_\As$ is a projector onto $\As$.

As proved in \cite[Prop. 1.5]{wolf2012quantum}, any conditional expectation on a unital algebra $\As$ with decomposition \eqref{eq:algblock} can be written in the form:
\begin{equation}
  \CE_{\As}(X)=U\left(\bigoplus_\ell \tr_{\Hc_{F,\ell}}[(I_{S,\ell}\otimes\tau_{F,\ell})(V_\ell X V^\dag_\ell)]\otimes I_{F,\ell}\right)U^\dag, 
  \label{eqn:cond_exp_blocks}
\end{equation}
with $V_\ell$ a linear operator from $\Hc$ onto $\Hc_{S,\ell}\otimes\Hc_{F,\ell}$ such that $V_\ell V_\ell^\dag=I_{S,\ell}\otimes I_{F,\ell},$ $V_\ell^\dag V_\ell = \Pi_{SF,\ell}\in\Bf(\Hc)$ the orthogonal projector onto $\Hc_{S,\ell}\otimes\Hc_{F,\ell}$, and $\tau_{F,\ell}$ are full-rank density operators on $\Hc_{F,\ell}.$  This shows that the conditional expectations on a given $*$-algebra are completely parametrized by the {\em factor states} $\tau_{F,\ell}.$ The state extension takes a similar form. 
\begin{proposition}
\label{prop:state_ext_blocks}
The dual of a $\CE_{\As}$ with decomposition \eqref{eqn:cond_exp_blocks} takes the form
\begin{equation}
  \SE_{\As}(X) = U\left(\bigoplus_\ell \tr_{\Hc_{F,\ell}}(V_\ell XV^\dag_\ell)\otimes \tau_{F,\ell}\right)U^\dag. 
  \label{eqn:state_ext_blocks}
\end{equation}
\end{proposition}
\begin{proof}
Let us start by recalling, from \cite[Lemma 1]{ticozzi2017alternating}, that for $\Hc = \bigoplus_\ell \Hc_\ell$ and $W = U(\bigoplus_\ell W_\ell)U^\dag$, we have $\tr(WY) = \sum_\ell \tr(W_\ell V_\ell Y V_\ell^\dag)$, where $V_\ell$ are the non-square isometries defined above. Through direct computation and using the fact that the dual of a partial trace is the tensor product with the identity on the traced-out subsystem, we then obtain
    \begin{align*}
        &\inner{X}{\CE_\As[Y]}_{HS}= \\
        &=  \inner{X}{U\left(\bigoplus_\ell \tr_{\Hc_{F,\ell}}[(I_{S,\ell}\otimes \tau_{F,\ell})(V_\ell Y V_\ell^\dag)]\otimes I_{F,\ell}\right)U^\dag}\\ 
        &= \sum_\ell \inner{V_\ell X V_\ell^\dag}{\tr_{\Hc_{F,\ell}}[(I_{S,\ell}\otimes\tau_{F,\ell})(V_\ell Y V_\ell^\dag)]\otimes I_{F,\ell}}\\
        &= \sum_\ell \inner{\tr_{\Hc_{F,\ell}}[V_\ell X V_\ell^\dag]}{\tr_{\Hc_{F,\ell}}[(I_{S,\ell}\otimes\tau_{F,\ell})(V_\ell Y V_\ell)]}\\
        &= \sum_\ell \inner{\tr_{F,\ell}[V_\ell X V_\ell^\dag]\otimes I_{F,k}}{(I_{S,\ell}\otimes\tau_{F,\ell})(V_\ell Y V_\ell^\dag)}\\
        &= \sum_\ell \inner{V_\ell^\dag\left[\tr_{\Hc_{F,\ell}}[V_\ell X V_\ell^\dag]\otimes \tau_{F,\ell}\right]V_\ell}{Y}\\
        &= \inner{\underbrace{U\left(\bigoplus_\ell \tr_{\Hc_{F,\ell}}[V_\ell X V_\ell^\dag]\otimes\tau_{F,\ell}\right)U^\dag}_{=\SE_\As[X]}}{Y}\\
        &= \inner{\SE_\As[X]}{Y}_{HS}
    \end{align*}
    Which concludes the proof.
\end{proof}

We thus have that ${\rm Im}(\SE_\As)$ contains a full-rank state of the form $\sigma=U\left(\bigoplus_\ell \frac{I_{S,\ell}}{\dim(\Hc_{S,\ell})}\otimes \tau_{F,\ell}\right)U^\dag$, 
with a block structure that is matching the one of $\As.$
By comparing the block representations \eqref{eqn:cond_exp_blocks} and \eqref{eqn:state_ext_blocks} with the one of $\sigma$, we have that
$ \SE_\As= \Dc_{\sigma}\CE_\As \Dc_{\sigma}^{-1},$ and in particular, for any $A\in\As,$ we have $ \SE_\As(A)= \Dc_{\sigma}(A).$
It follows that 
\begin{align*}
{\rm Im} (\SE_\As)&=\fix (\SE_\As)=\Dc_{\sigma}({\rm Im} (\CE_\As))=\Dc_{\sigma}(\As)\\
&= U \left(\bigoplus_\ell \Bf(\Hc_{S,\ell})  \otimes\tau_{F\ell}\right) U^\dag.
\end{align*}

\begin{remark}
    Note that $\forall \rho\in\Df(\As)$, $\SE_\As(\rho) = \Dc_\sigma(\rho)\in\Df(\Hc),$ and is  a thus a preserved state for $\CE_\mathscr{A},$ which motivates the name of state extension: it is a map that extends a state in the subalgebra to a general one that is preserved by its dual conditional expectation.
\end{remark}

Using the form of the full rank state $\sigma$ given above, it is easy to see that indeed $\CE_\As$ is self-adjoint with respect to any modified inner product of the form
$\inner{X}{Y}_{\Dc_\sigma}=\tr(X^\dag \sigma^\um Y\sigma^{\um})$ with $\Dc_{\sigma}(\cdot) \equiv \sigma^{\um}\cdot\sigma^\um$, see e.g. \cite[Proposition 2]{ticozzi2017alternating}. It then follows that $\CE_\As$ is an orthogonal projection, with respect to the Hilbert-Schmidt inner product, if and only if all $\tau_{F,\ell}$ are completely mixed states on their supports. This is equivalent to $\SE_\As$ being unital.

Similarly, we have that $\SE_\As$ is self-adjoint with respect to the modified inner products
$\inner{X}{Y}_{\Dc_\sigma^{-1}}=\tr(X^\dag \sigma^{-\um} Y\sigma^{-\um})$ with $\Dc_\sigma^{-1}(\cdot) \equiv\sigma^{-\um}\cdot\sigma^{-\um}$.

\subsection{ Unital factorizations of conditional expectations and their duals} 
Leveraging the Wedderburn decomposition of $\As$ we can reduce the size of the representation of a $*$-subalgebra by avoiding the repeated blocks. In particular, given an algebra $\As\subseteq\Bs$ with Wedderburn decomposition as in  \eqref{eq:algblock} we can observe that it is isomorphic to $\check{\As} = \bigoplus_\ell \Bf(\Hc_{S,\ell})$, $\check{\As}\subseteq \Cb^{m\times m}$ with $m=\sum_\ell \dim(\Hc_{S,\ell})$. 
We next show that this reduction in representation is possible using CP unital or CPTP maps that factorize the conditional expectation or its dual, respectively. This result, albeit simple, is key to our aim, as it allows us to construct reductions of CP dynamics that remain CP.

\begin{theorem}\label{thm:factorizations}
Let $\As\subseteq\Bs$ be an unital $*$-subalgebra with decomposition as in \eqref{eq:algblock}. Define $\check \As=\bigoplus_\ell \Bs(\Hc_{S,\ell}).$ Then for any conditional expectation $\CE_\As$ and state extension $\SE_\As$ there exist (non-square) factorizations:
\begin{equation}
\CE_\As = \eo \ro, \quad \SE_\As = \es \rs
\end{equation}
where 
\begin{align*}
    \eo&:\check\As\rightarrow \Bs,&\ro&:\Bs\rightarrow \check\As, \\
    \es&:\check\As\rightarrow \Bs,&\rs&:\Bs\rightarrow \check\As,
\end{align*}
with $\eo, \ro$ CP and unital maps while $\es, \rs$ are CPTP, and so that for any $X\in\Hf(\Hc)$, $\rho\in\Df(\Hc)$ it holds:
\begin{equation}
    \tr(\CE_\As(X)\rho)=\tr(\ro(X)\rs(\rho))=\tr(X\SE_\As(\rho)).
\end{equation}
\end{theorem}

\begin{proof}
Define the  CP unital linear map $\Phi:\As\rightarrow \check\As$ as:
\begin{equation*}
    \Phi(A)=\Phi\left(U\left(\bigoplus_\ell A_{S,\ell}\otimes I_{F,\ell}\right)U^\dag\right)=\bigoplus_\ell A_{S,\ell}.
\end{equation*}
Defining $V_\ell$ as in the previous subsection, we can extend the action of $\Phi$ to the full space as:
$$
    \Phi(X)= \bigoplus_\ell \frac{\tr_{\Hc_{F,\ell}}(V_\ell XV_\ell^\dag)}{\dim(\Hc_{F,\ell})}.
$$
Similarly, for any  $\check A=\bigoplus_\ell A_{S,\ell} \in\check{\As},$ we have
$$
  \Phi^{-1}\left(\check A\right) = \Phi^{-1}\left(\bigoplus_\ell A_{S,\ell}\right) = U\left(\bigoplus_\ell \check A_{S,\ell}\otimes I_{F,\ell}\right)U^\dag.
$$
so that $\Phi^{-1}\Phi=\mathcal{I}_\As$, the identity super operator over the algebra. Notice that also $\Phi^{-1}$ is unital.
We can then exploit this reduction in our problem by noting that
$\CE_\As=\Phi^{-1}\Phi\CE_\As,$
and define the following (non-square) unital factorization for $\CE_\As$:
$$\ro=\Phi\CE_\As,\quad \eo=\Phi^{-1},$$
so that $\CE_\As=\eo \ro.$ Unitality can be verified directly from the definition of $\Phi$ and its inverse.
In block decomposition, with the notation introduced above, we have:
$$\ro(X)=\bigoplus_\ell \tr_{\Hc_{F,\ell}}\left(I_{S,\ell}\otimes\tau_{F,\ell}(V_\ell XV_\ell^\dag)\right).$$ 
We can use a similar representation for $\SE_\As$ as well, which follows directly from duality. In this case we use the dual CPTP map
\begin{equation}\label{eqn:reduction}
    \rs(X)=\eo^\dag(X)=(\Phi^{-1})^{\dag}(X)=\bigoplus_\ell \tr_{\Hc_{F,\ell}}(V_\ell X V_\ell^\dag)
  \end{equation}
and, for $\check A=\bigoplus_\ell A_{S,\ell},$ we can define
\begin{equation}\label{eqn:injection}
    \es(\check A)=\ro^\dag(\check A)=U\left(\bigoplus_\ell A_{S,\ell} \otimes\tau_{F,\ell}\right)U^\dag.\end{equation}
The last equation follows from $\CE_\As=\SE^\dag_\As$ and the factorization we just defined.
\end{proof}

Notice that the explicit form of the $\rs$ and $\es$ is derived in the proof, see \eqref{eqn:reduction},\eqref{eqn:injection}.  
 Leveraging the the latter, we can obtain the reduced QSO model as described in the following instrumental result.

\begin{proposition}
\label{prop:reduction}
    Let $(\Bs,\Ys,\Ac,\Cc, \Sf)$ be a QSO model, let  $\As\subseteq\Bs$ be a sub-algebra and let $\SE_{\As}$ be a state extension associated to $\As,$  such that, for all $t\geq 0$ and $\rho_0\in\Sf,$ we have
    \begin{equation}
       \Cc\Ac^{t}[\rho_0] = \Cc\SE_\As(\SE_\As\Ac\SE_\As)^t \SE_\As[\rho_0]. 
       \label{eqn:reduction1}
    \end{equation}
     Then, if $\SE_\As$ admits CPTP factorization $\SE_\As = \es\rs,$ the reduced QSO $(\check{\As},\Ys,\check{\Ac},\check{\Cc},\check{\Sf})$ with $\check{\As} = {\rm Im}(\rs)$, $\check{\Ac} = \rs\Ac\es$, $\check{\Cc} = \Cc\es$ and $\check{\Sf} = \rs \Sf$, along with $\Phi=\rs$ solves Problem \ref{prob:single_time}.
\end{proposition}
\begin{proof}
    Let us start by proving that the reduced model is a QSO model. From Theorem \ref{thm:factorizations} we have that $\check{\As}$ is a $*$-algebra and the maps $\es$ and $\rs$ are CPTP. Then, since $\Ac$, $\rs$ and $\es$ are CPTP, $\check{\Ac}$ is also CPTP and $\check{\Sf}$ is a set of density operators. 
    Using the definitions of the Theorem \ref{thm:factorizations}, one can directly verify that $\rs\es = \Ic_{\check{\As}}$ the identity super operator over $\check{\As}$, we have that \(\Cc\es (\rs\Ac\es)^{t} \rs[\rho_0] = \Cc\SE_\As (\SE_\As\Ac\SE_\As)^t \SE_\As[\rho_0]\) for all $\rho_0\in\Sf$ and $t\geq0$ proving that this is a solution for Problem \ref{prob:single_time} with $\Phi=\rs$.
\end{proof}

This result allows us to focus on finding a CPTP projection $\SE_\As$ such that equation \eqref{eqn:reduction1} holds. 
The core idea follows naturally from the last results and two observations: (i) the reachable space $\Rs$ and  $\Ns^\perp$ represent subspaces that support reductions for the dynamics;  (ii) composing a CPTP injection map with a CPTP dynamics and a CPTP reduction map amounts to a reduced CPTP dynamics. 
Hence given a reachable subspace or the subspace orthogonal to the non-observable subspace we shall extend them to an an operator subspace that is the image of a CPTP projection $\SE_\As$. The latter can then factorize in a CPTP injection and reduction pair.
We are then left with one open problem: how to construct an operator space that is the image of a CPTP projection $\SE_\As$, and is of minimal dimension. This shall be the focus of the next section.

\section{Building minimal algebras that admit\\ CPTP projections}
\label{sec:minimal_distorted_algebras}

\subsection{Distorted algebras and why we need them}
Given Proposition \ref{prop:reduction} above, we know we can obtain reduced CPTP models if we can find a suitable state extension. We next further investigate what kind of set ${\rm Im} (\SE_\As)$ is, and why it is key to our reduction task.
We start by recalling (see Section \ref{sec:Conditional_expectations}) that $${\rm Im} (\SE_\As)=\fix (\SE_\As)=\Dc_{\sigma}({\rm Im} (\CE_\As))=\Dc_{\sigma}(\As).$$
Sets of this form correspond to fixed point sets of general CPTP maps - not just state extensions: in fact for any CPTP map $\Ec$ there is a state extension $\bar \Ec=\lim_{T\rightarrow \infty}\frac{1}{T+1}\sum_{i=0}^T\Ec^i$ that projects onto its fixed points.

A set of the form $\Dc_{\sigma}(\As)$ with $\sigma\in\mathfrak{H}(\mathcal{H})$ is also called 
{\em $\sigma$-distorted algebra} \cite{ticozzi2017alternating,blume2010information,johnson2015general}. It is immediate (through the block-decomposition representation and using the properties of tensor product) to verify that $\Dc_{\sigma}(\As)$ is closed with respect to linear combination, adjoint, and the weighed product $X\cdot_{\sigma} Y=X \sigma^{-1} Y.$ 
Let us fix some notation: for a subset $\Sc$ of $\mathbb{C}^{n\times n}$, we denote by ${\alg}\,\Sc$ the minimal $*$-subalgebra that contains the set $\Sc,$ and similarly by ${\alg}_{\sigma}\,\Sc$ the minimal $\sigma$-distorted algebra that contains $\Sc$, closed with respect to the product $\cdot_\sigma.$ 
On the other hand, any $*$-closed algebra with respect to a modified $\cdot_\sigma$ can be shown to have the form $\Dc_{\sigma}(\As)$ for some (standard) $*$-algebra $\As$ \cite{MyThesis}, thus motivating the name: these can always be obtained as a ``distortion'' of a standard algebra.

However, not all distorted algebras are images of state extensions: this is the case if and only if the distortion state is fixed for some conditional expectation on $\As.$ Takesaki's  work \cite{takesaki1979theory} employs modular theory to characterize such states, while the equivalence  will be proved explicitly in Theorem \ref{thm:takesaki_extended}. Our reduction strategy hinges on the existence of a CPTP projection onto a distorted algebra, so we need to determine a valid density operator,  say $\sigma$,  with respect to which the relevant operator subspace,  say $\Vs$, can be closed to a distorted algebra $\alg_\sigma \Vs $ that is the image of a CPTP projection, but not only: we want the state $\sigma$ to produce the {\em minimal} such distorted algebra, i.e. given all the states $\sigma_j$ such that all $\alg_{\sigma_j}\Vs$ are images of CPTP projections, we would like to find the particular $\sigma_j^\star$ such that $\dim(\alg_{\sigma_j^\star} \Vs)$ is minimal.

A natural question comes to mind: do we really need to consider distorted algebras with respect to general states? Given a set a generators $\cal S$, the standard algebra $\alg({\cal S})$ is a particular distorted algebra (with respect to the completely mixed state), and as such admits a state extension. The following example shows how the choice of this algebra may lead to non-minimal reductions.

\begin{example}
    Consider a CPTP map $\Ac$ with a fixed point $\bar{\rho}$, i.e. $\Ac(\bar{\rho})=\bar{\rho}$. Consider as an initial condition $\Sf=\{\bar{\rho}\}$. The reachable space generated by the model $\rho(t+1) = \Ac(\rho(t))$ is then $\Rs = \Span \{\bar{\rho}\}$. When computing $\alg(\Rs)$ we obtain $\alg(\Rs) = \Span\{\Pi_i\}$ where $\Pi_{i}$ are orthogonal projectors onto the eigenspaces of $\bar{\rho}$ and thus $\dim(\alg(\Rs))$ is the number of distinct eigenvalues of $\bar{\rho}$.

    However, we can observe that choosing $\check{\rho}(0)=1$, $\check{\Ac} = 1$ and $\es(\check{\rho}) = \bar{\rho}\check{\rho}$ we have that the model $\check{\rho}(t+1) = \check{\Ac}[\check{\rho}(t)]$  and $\rho(t) = \es(\check{\rho}(t)),$ with initial condition $\check{\rho}(0)=1$, provides the correct trajectory at all times, i.e. $\rho(t)=\bar{\rho}$ for all $t\geq0$, and, since the dimension of this model is 1, it must also be minimal. 
    This minimal model can be obtained from the reachable subspace by closing $\Rs$ to a distorted algebra using $\bar\rho$ to define the modified product, trivially obtaining
    $\alg_{\bar{\rho}}(\Rs)=\Rs$. This proves that, in this case, $\Rs$ is in fact a distorted algebra of dimension $1$. Furthermore $\tr[\cdot]\bar{\rho}$ is a CPTP projection onto $\Rs$ that can be factorized in $\Rc(\cdot) = \tr[\cdot]$ and $\Jc(\cdot) = \cdot \bar{\rho},$ which lead to the minimal model introduced above. 
\qed
\end{example}
This simple example shows that, although closing $\Rs$ to an algebra provides a reduced quantum model, it is possible that the reduced model we obtain in this way may not be minimal in size. Closing $\Rs$ to a distorted algebra can instead lead to smaller reduced models. This fact will be formally proved in the following.

The rest of the section will be devoted to finding a distorted algebra that contains a given operator subspace and is the image of a CPTP projector that leads to an optimal reduction. This is achieved by leveraging a number of existing and new results on conditional expectations.

\subsection{$\mathscr{A}$-factorized states and conditional expectations}

We first need to introduce a new concept, that captures density operators and algebras sharing compatible block-diagonal structures.

\begin{definition}
    Let $\As\subseteq\Bf(\Hc)$ be a $*$-algebra. 
    We say that an operator $\sigma\in\Bf(\Hc)$ is \em $\As$-factorized, and write $\sigma\comp\As,$ if $\sigma$ can be written as the product $\sigma = \sigma_A\sigma_C$ for some $\sigma_A\in\As$ and $\sigma_C\in\As'$.

   By extension, we say that $\sigma$ is $\As_\mu$-factorized and write $\sigma\comp\As_\mu$ if it is compatible with the corresponding un-distorted algebra, i.e. $\sigma\comp\Dc_\mu^{-1}(\As_\mu)$.
\end{definition}

Recalling the Wedderburn decomposition of the algebra and its commutant given in equations \eqref{eq:algblock} and \eqref{eq:commblock}, we have that $\sigma\sim\As$ if and only if 
\begin{equation}
    \sigma = U\left(\bigoplus_\ell \sigma_{S,\ell}\otimes \tau_{F,\ell} \oplus \zero_R \right)U^\dag,
\end{equation}
for some $\sigma_{S,\ell}\in\Df(\Hc_{S,\ell})$, and $\tau_{F,\ell}\in\Df(\Hc_{F,\ell})$.
Hence, we have that, given an algebra $\As$, any $\sigma\in\As$, $\sigma\in\As'$ and in particular $\sigma\in\zentrum(\As),$ satisfy $\sigma\comp\As$. The proof follows by comparing the block structures of $\As$ and $\As'$.

The notion of $\As$-factorized operators helps us harness the block structure provided by the Wedderburn decomposition and is of fundamental importance because it provides an intuitive interpretation of the necessary and sufficient conditions for the existence of a CPTP projection on a distorted algebra provided by Takesaki theorem in the finite-dimensional case. This connection is explicitly proved next.

\begin{theorem}[Takesaki theorem and $\As$-factorizability]
\label{thm:takesaki_extended}
    Let $\sigma$ be a full-rank positive-definite operator, $\As\subseteq\Bf(\Hc)$ a unital $*$-algebra and $\As_\sigma=\Dc_\sigma(\As)$ a $\sigma$-distorted algebra. The following conditions are equivalent:
    \begin{enumerate}
        \item $\exists$ $\CE_\As$ that preserves ${\sigma}$;
        \item $\exists$ $\SE_\As$ such that $\im(\SE_\As)= \As_\sigma$;
        \item $\sigma\comp\As_\sigma$.
    \end{enumerate}
\end{theorem}
\begin{proof}
  Takesaki theorem \cite{TAKESAKI1972306} implies that 1) above is equivalent to $\As_\sigma$ being invariant for the modular action defined as $\Mc_\sigma(\cdot)\equiv \sigma^\um \cdot \sigma^{-\um}$, i.e. $\Mc_\sigma(\As_\sigma)\subseteq\As_\sigma$. The proof of the equivalence of 1) and 2) in the finite-dimensional case can be found in \cite[Theorem 9.2]{petz2007quantum} and \cite[Theorem 3 and 4]{johnson2015general}.

   We first show that condition 3), implies $\Mc_\sigma(\As_\sigma)\subseteq\As_\sigma.$ 
   Assume $\As$ has Wedderburn decomposition $ \mathscr{A} =  U\left(\bigoplus_\ell \Bf(\Hc_{S,\ell}) \otimes I_{F,\ell}  \right)U^\dag$: thus $\sigma\comp \mathscr{A}$ (or equivalently $\sigma\comp \mathscr{A}$) implies $\sigma = U\left(\bigoplus_\ell \sigma_{S,\ell} \otimes \tau_{F,\ell}  \right)U^\dag$. Then, we have \begin{align*}
   \Mc_\sigma(\As_\sigma) &= U\left(\bigoplus_\ell \sigma_{S,\ell}^{\um}\Bf(\Hc_{S,\ell})\sigma_{F,\ell}^{-\um} \otimes \tau_{F,\ell}^\um \tau_{F,\ell}\tau_{F,\ell}^{-\um}  \right)U^\dag\\ &\subseteq U\left(\bigoplus_\ell \Bf(\Hc_{S,\ell}) \otimes \tau_{F,\ell} \right)U^\dag = \As_\sigma,
   \end{align*} 
   and thus $\As_\sigma$ is $\Mc_\sigma$-invariant.
    
    To conclude, we prove that 2) implies 3). By hypothesis we have that there exists a CPTP map $\SE_{\As}$, such that $\im(\SE_\As) = \As_\sigma$.  From Proposition \ref{prop:state_ext_blocks}, we than have that such map has the form \eqref{eqn:state_ext_blocks}, hence, its image has the structure $\im(\SE_\As) = U\left(\bigoplus_\ell \Bf(\Hc_{S,\ell})\otimes \tau_{F,\ell}\right)U^\dag = \As_\sigma$. From Proposition \ref{prop:distorted_algebra} we have that $\sigma\in\As_\sigma$, and thus $\sigma$ must have the block structure $\sigma = U\left(\bigoplus_\ell \sigma_{S,\ell} \otimes \tau_{F,\ell} \right)U^\dag $ and is thus $\As$-factorized with $\As = \Dc_{\sigma}^{-1}(\As_\sigma) = U\left(\bigoplus_\ell \Bf(\Hc_{S,\ell})\otimes I_{F,\ell}\right)U^\dag$. 
\end{proof}

We have thus shown that the existence of a CPTP projection onto a distorted algebra $\Dc_\sigma(\As)$ is equivalent to the fact that $\sigma$ is $\As_\sigma$-factorized, or equivalently that exists a conditional expectation on $\As$ preserving $\sigma.$
Given an operator space $\Vs$ we now want to understand how to choose $\sigma\in\Bf(\Hc)$ so that $\sigma\sim\alg_\sigma\Vs,$ and the latter algebra is minimal.

A first possibility is to pick a full rank density operator $\sigma$  inside $\alg \Vs$. The next result provides a necessary and sufficient condition for an operator $\sigma\in\alg\Vs$ to be such that $\sigma\comp\alg_\sigma\Vs$ and shows that, for such states, closing with respect to a modified product never leads to larger algebras.
\begin{theorem}\label{thm:minimaldistorted}
    Consider an operator space $\Vs\subseteq\Bf(\Hc)$ with full support and consider $\sigma\in\alg\Vs$. 
    Let then  
    \begin{align*}
        \alg\Vs &= U\left(\bigoplus_\ell \Bf(\Hc_{S,\ell})\otimes I_{F,\ell}\right)U^\dag,\\
        \sigma &= U\left(\bigoplus_\ell \sigma_{S,\ell} \otimes I_{F,\ell} \right)U^\dag
    \end{align*}
    be the Wedderburn decomposition of $\alg\Vs$ and the structure of $\sigma
     $ with $\sigma_{S,\ell}\in\Bf(\Hc_{S,\ell})$, $\sigma_{S,\ell}>0$. Then: 
\begin{enumerate}
        \item \[\Vs\subseteq\alg_\sigma\Vs\subseteq\alg\Vs;\]
        \item  \[\alg_\sigma\Vs = U\left(\bigoplus_\ell \Dc_{\sigma_{S,\ell}}(\As_{S,\ell})\otimes I_{F,\ell}\right) U^\dag\]       where  $\As_{S,\ell}\subseteq\Bf(\Hc_{S,\ell})$ are (sub)algebras;
        \item $\sigma\comp\alg_\sigma \Vs$ if and only if $\As_{S,\ell} = \Bf(\Hc_{S,\ell})$ $\forall l$.
    \end{enumerate}
\end{theorem}
\begin{proof} 

    Given a basis $\{V_i\}_{i=0,\dots}$ for $\Vs$ we have that $V_i = U\left(\bigoplus_\ell V_{S,\ell}^i\otimes I_{F,\ell}\right)U^\dag $ with $\alg\{V_{S,\ell}^i\}_{i=0,\dots}=\Bf(\Hc_{S,\ell})$ for all $\ell$. Then we have that $\Dc_\sigma^{-1}(\Vs) = \Span\{W_i\}$ with \(W_i = U\left(\bigoplus_\ell \sigma_{S,\ell}^{-\um}V_{S,\ell}^i\sigma_{S,\ell}^{-\um} \otimes I_{F,\ell}\right)U^\dag\). Observing that linear combinations, products, and adjoints of elements that have the structure of $W_i$ maintain the same block-diagonal structure, and using Proposition \ref{prop:distorted_algebra}, i.e. $\alg_\sigma\Vs = \Dc_\sigma(\alg(\Dc_\sigma^{-1}(\Vs)))$, we have $ \alg_\sigma\Vs = U\left(\bigoplus_\ell \sigma_{S,\ell}^{\um}  \As_{S,\ell}   \sigma_{S,\ell}^{\um}\otimes I_{F,\ell}\right)U^\dag$ where $\As_{S,\ell}\equiv \alg(\sigma_{S,\ell}^{-\um}V_{S,\ell}^i\sigma_{S,\ell}^{-\um})\subseteq\Bf(\Hc_{S,\ell})$, concluding the proof of statement 2). 
    
    From $\As_{S,\ell}\subseteq\Bf(\Hc_{S,\ell})$ for all $\ell$ it immediately follows that  $\Dc_{\sigma_{S,\ell}}(\As_{S,\ell})\subseteq\Bf(\Hc_{S,\ell})$ for all $\ell$ and thus we can conclude that $\Vs\subseteq\alg_\sigma\Vs\subseteq\alg\Vs$, proving statement 1).
    
    To prove the third statement and conclude the proof we can observe that for every block $\ell$, we have $\Vs_{S,\ell}\equiv \Span\{V_{S,\ell}^i\}_{i}$ such that $\alg\Vs_{S,\ell} = \Bf({\Hc_{S,\ell}})$ and we can thus apply Lemma \ref{lem:little_blocks} on every block $\As_{S,\ell}$ for all $\ell$.
\end{proof}

The necessary and sufficient condition of statement 3) is however difficult to verify in general. Nonetheless, we next show that any (full-rank) operator in the center of $\alg\Vs$,  $\sigma\in\zentrum(\alg\Vs)$ guarantees that $\sigma\sim\alg_\sigma\Vs$. This is summarized in the following Corollary. 
\begin{corollary}
    Consider an operator space $\Vs\subseteq\Bf(\Hc)$ with full support. Let us consider a  positive-definite operator in the center of $\alg\Vs$, i.e. $\sigma\in\zentrum(\alg\Vs)$. Then it holds that $\Vs\subseteq\alg_\sigma\Vs\subseteq\alg\Vs$ and $\sigma\comp\alg_\sigma\Vs$.
\end{corollary}
\begin{proof}
    By Theorem \ref{thm:minimaldistorted}, we have $\Vs\subseteq\alg_\sigma\Vs\subseteq\alg\Vs$, since $\zentrum(\alg\Vs)\subseteq\alg\Vs$. To prove the fact that $\sigma\comp\alg_\sigma\Vs$ we can simply observe that for $\sigma\in\zentrum(\alg\Vs)$ we have that $\sigma = U\left(\bigoplus_\ell \lambda_\ell I_{S,\ell} \otimes I_{F,\ell} \right)U^\dag$ with $\lambda_\ell\in\Rb$ which implies that $\As_{S,\ell} \equiv \alg(I_{S,\ell}V_{S,\ell}^iI_{S,\ell})=\Bf(\Hc_{S,\ell})$ for all $\ell$ and thus satisfies the third condition of Theorem \ref{thm:minimaldistorted}. 
\end{proof}

This corollary shows that picking $\sigma\in\zentrum(\alg\Vs)$ not only guarantees that $\alg_\sigma\Vs$ is the image of a CPTP projection, but also $\dim(\alg_\sigma\Vs)\leq\dim(\alg\Vs),$ thus possibly allowing for a larger reduction than that provided by $\alg \Vs$. This is confirmed by the next example.

\begin{example}
Let $\sigma_k$ with $k=x,y,z$ indicate the Pauli matrices and $\sigma_0$ indicate the identity.
Let us define a positive-definite operator $\mu = 2\sigma_0 + \sigma_z$ and consider the following operator space $\Vs = \Span\{\mu\otimes\sigma_k,\, k=x,y,z\}$. Then $ \alg\Vs = \Span\{\sigma_j\otimes\sigma_k, \sigma_0\otimes\sigma_k, \, k=0,x,y,z\} \simeq \Bf(\Cb^2)\oplus \Bf(\Cb^2)$, with $\dim(\alg\Vs)=8$ and we can observe that $\zentrum(\alg\Vs) = \Span\{\sigma_j\otimes\sigma_0, \, j=0,z\}$. 

Let us then consider $\sigma = \mu\otimes\sigma_0 \in\zentrum(\alg\Vs)$, then we have that $\Dc_\sigma^{-1}(\Vs) = \Span\{\sigma_0\otimes\sigma_k, \, k=x,y,z\}$. This implies that $\alg(\Dc_\sigma^{-1}(\Vs)) = \Span\{\sigma_0\otimes\sigma_k, \, k=0,x,y,z\} =  \sigma_0\otimes\Bf(\Cb^2)$ and thus $\alg_\sigma\Vs = \mu\otimes \Bf(\Cb^2)$ with  $\dim(\alg_\sigma\Vs)=4$ and $\alg_\sigma\Vs$ is $\Mc_\sigma$-invariant or, in other words, $\sigma\comp\alg_\sigma\Vs$.
\end{example}

\subsection{Distorted algebras of minimal dimensions}

Up to this point, we only showed that choosing $\sigma\in\zentrum(\alg\Vs)$ is such that $\sigma\sim\alg_\sigma\Vs$ and $\dim(\alg_\sigma\Vs)\leq\dim(\alg\Vs)$, but one could ask if there exist other states $\xi\in\Bf(\Hc)$ such that $\xi\comp\alg_\xi(\Vs),$ and such that $\dim(\alg_\xi\Vs)\leq\dim(\alg_\sigma\Vs)$ for all $\sigma\in \zentrum(\alg\Vs)$ We  now show that this is not possible and that the optimal reduction can be always obtained by considering an operator $\sigma$ in the center $\zentrum(\alg\Vs)$.

\begin{theorem}[Minimal distorted algebra]
  Consider an operator subspace $\Vs\subseteq\Bf(\Hc)$ and a positive-definite operator $V\in\Vs$, $V=V^\dag>0$. Define $$\sigma \equiv \Pi_{\zentrum(\alg\Vs)}\left[V\right]\in\zentrum(\alg\Vs)$$ where $\Pi_{\zentrum(\alg\Vs)}$ is the orthogonal projection onto the center ${\zentrum(\alg\Vs)}$. Then, $\alg_\sigma\Vs$ is the distorted algebra of minimal dimension that contains $\Vs$ and such that $\sigma\sim\alg_\sigma\Vs$.
\end{theorem}
\begin{proof}
    Consider a full rank density operator $\xi\in\Df(\Hc)$ and define $\As \equiv \alg(\Dc_\xi^{-1}(\Vs))$. Being $\As$ a $*$-algebra, it admits a Wedderburn decomposition: $\As = U\left(\bigoplus_\ell \Bf(\Hc_{S,\ell})\otimes I_{F,\ell}\right)U^\dag$. Let $\{V_i\}$ be a set of generators for $\Vs$, and for all $i$ define $W_i=\Dc_\xi^{-1}(V_i).$ From the structure of $\As$ we also have that $W_i = U\left(\bigoplus_\ell W_{S,\ell}^i\otimes I_{F,\ell} \right)U^\dag,$ and that $\alg\{W_{S,\ell}^i\} = \Bf(\Hc_{S,\ell})$ for all $\ell$. Moreover, assuming $\xi\comp\alg(\Dc_\xi^{-1}(\Vs))$ we have that $\xi$ has the structure $\xi = U\left(\bigoplus_\ell \xi_{S,\ell}\otimes\tau_{F,\ell} \right)U^\dag$. Combining these two observations we have ${\alg_\xi\Vs}= U\left( \bigoplus_\ell \xi_{S,\ell}^{\um} \Bf(\Hc_{S,\ell}) \xi_{S,\ell}^{\um} \otimes \tau_{F,\ell} \right)U^\dag.$ 
    
    By definition, we have that $\Vs\subseteq\alg_\xi\Vs$. This, with the structure of the basis of $\Dc_\xi^{-1}(\Vs)$ we just described, implies that for the basis elements $\Vs = \Span\{V_i\}$ we can write 
    \[V_i = U\left(\bigoplus_\ell \underbrace{\xi_{S,\ell}^\um W_{S,\ell}^i\xi_{S,\ell}^\um}_{\equiv V_{S,\ell}^i} \otimes \tau_{F,\ell} \right)U^\dag. \]
    We can then compute the algebra $\alg\Vs$. Since the $V_i$ share a common block-diagonal structure, and the latter is invariant for sum, multiplication and adjoint, we have that $\alg\Vs = U\left(\bigoplus_\ell \alg\{V_{S,\ell}^i\}_i\otimes\alg\{\tau_{F,\ell}\}\right)U^\dag$. It follows that its center has the following structure: $\zentrum(\alg\Vs) = U\left(\bigoplus_\ell \zentrum(\alg\{V_{S,\ell}^i\})\otimes \zentrum(\alg\{\tau_{F,\ell}\})\right)U^\dag$.
    
    Thus any $V\in\Vs$ is of the form $V = U\left(\bigoplus_\ell V_{S,\ell}\otimes \tau_{F,\ell} \right)U^\dag$ and, since $\tau_{F,\ell}\in\zentrum(\alg(\tau_{F,\ell}))$, we have $$\sigma = \Pi_{\zentrum(\alg\Vs)}[V] = U\left(\bigoplus_\ell \sigma_{S,\ell}\otimes \tau_{F,\ell} \right)U^\dag$$ where $\sigma_{S,\ell}\in\zentrum(\alg\{V_{S,\ell}^i\})$. Moreover, we have that 
    $\Dc_\xi^{-1}(\sigma)\in\alg(\Dc_\xi^{-1}(\Vs)),$ since 
    \begin{align*}
        \Dc_\xi^{-1}(\sigma) &= U\left( \bigoplus_\ell \xi_{S,\ell}^{-\um} \sigma_{S,\ell} \xi_{S,\ell}^{-\um} \otimes \tau_{F,\ell}^{-\um} \tau_{F,\ell} \tau_{F,\ell}^{-\um} \right) U^\dag\\ &= U\left( \bigoplus_\ell \xi_{S,\ell}^{-\um} \sigma_{S,\ell} \xi_{S,\ell}^{-\um} \otimes I_{F,\ell} \right) U^\dag
    \end{align*}
    and $ \alg(\Dc_\xi^{-1}(\Vs))=U\left( \bigoplus_\ell \Bf(\Hc_{S,\ell}) \otimes I_{F,\ell} \right) U^\dag$. This shows that $\sigma$ satisfies the hypothesis of Lemma \ref{lem:distorted_algebra_inclusion} and we are guaranteed that \begin{equation}\label{eq:first}
    \alg(\Dc_\sigma^{-1}(\Vs))\subseteq\alg(\Dc_\xi^{-1}(\Vs)). \end{equation} Next, observe that 
        \begin{align}  &\Dc_\sigma^{-1}\left(\alg_\xi\Vs\right)=\nonumber\\ 
        &= U\left(\bigoplus_\ell \sigma^{-\um}_{S,\ell}\xi_{S,\ell}^\um \Bf(\Hc_{S,\ell}) \xi_{S,\ell}^\um \sigma^{-\um}_{S,\ell} \otimes \tau_{F,\ell}^{-\um}\tau_{F,\ell}\tau_{F,\ell}^{-\um} \right)U^\dag\nonumber\\
        & = U\left(\Bf(\Hc_{S,\ell})\otimes I_{F,\ell}\right)U^\dag = \alg(\Dc_\xi^{-1}(\Vs)).\label{eq:second}
    \end{align}
    Combining \eqref{eq:first} and \eqref{eq:second} we have that {\em for any $\xi\comp\alg_\xi\Vs,$} we have $\alg\left(\Dc_\sigma^{-1}(\Vs)\right) \subseteq \alg(\Dc_\xi^{-1}(\Vs)) = \Dc_\sigma^{-1}\left(\alg_\xi\Vs\right)$. Applying $\Dc_\sigma$ on both sides of the previous equation we get:
    $\Dc_\sigma\left(\alg\left(\Dc_\sigma^{-1}(\Vs)\right)\right) \subseteq \alg_\xi\Vs,$ which by Proposition \ref{prop:distorted_algebra} is equivalent to  $\alg_\sigma\Vs\subseteq \alg_\xi\Vs,$ for any $\xi$ such that $\xi\sim\alg_\xi\Vs$. This implies that $\alg_\sigma\Vs$ is the minimal distorted algebra containing $\Vs.$
\end{proof}

It follows from the proof of the above Theorem that the choice of the operator $V$ in the statement does not affect the algebra we obtain. Moreover, while the choice above yields the optimal (minimal) distorted algebra containing $\Vs$ that admits a CPTP projection, it is possible that there exists an operator $\xi\in\Bf(\Hc)$ such that $\alg_\xi\Vs$ is smaller in dimension with respect to all $\alg_\sigma\Vs$ with $\sigma\in\zentrum(\alg\Vs)$. However, by Theorem \ref{thm:takesaki_extended} such $\alg_\xi\Vs$ would not be the image of a CPTP projection (state extension) and thus could not be used to obtain a CPTP reduced model via conditional expectations. 

\section{Reductions to QSO models}
\label{sec:reduction_algorithms}

The algorithm we propose is divided into two steps, each providing a solution to Problem \ref{prob:single_time}, but the solution provided by a single step is optimal only under certain conditions that we will discuss later. In the following subsections, we will start by discussing one step at a time and then show how to combine the two steps to obtain the full algorithm.

\begin{remark}
    In the remainder of this work, we will only focus on algebras with full support, i.e. $\Hc_R=0$. This is always possible since, if computing an algebra one finds that it does not have full support, one can always first restrict the QSO models and their analysis onto the support of the algebra itself. In fact, the set of operators over the support of the algebra is a $*$-algebra that allows for a state extension. 
A case in which this happens will be discussed in the examples section.
\end{remark}

\subsection{Projection onto the reachable algebra}

First, we need to define two objects that  play a key role in the algorithm and in the subsequent Theorem. The state $\bar{\rho}$ is defined as the weighted sum of the trajectories that generate the reachable space, i.e. 
\begin{equation}
    \bar{\rho}\equiv \frac{1}{|\Sf| n^2}\sum_{\rho_0\in\Sf}\sum_{t=0}^{n^2} \Ac^{t}[\rho_0]
    \label{eq:rho_bar}
\end{equation}
where $n=\dim(\Hc)$. Let us then denote $\zentrum_R \equiv \zentrum(\alg(\Rs))$ and $\Pi_\zentrum$ the projector onto the center $\zentrum_R$. The projection of $\bar{\rho}$ onto $\zentrum_R$  is then denoted by $\sigma\equiv \Pi_\zentrum[\bar{\rho}]$. 

We next introduce the \textit{reachable algebra}, defined as \[\Ds = \alg\left(\Dc_\sigma^{-1}(\Rs)\right).\] We shall see in Theorem \ref{thm:reachable-projection} that $\Ds$, does not depend on the state $\sigma$ as long as $\sigma$ is chosen as the projection onto the center $\zentrum$ of a full support state. 

We now present Algorithm \ref{algo:reachable}, which provides a partial solution to Problem \ref{prob:single_time}, and prove that the reachable algebra is the minimal algebra that allows the reduced model to reproduce the state dynamics for all possible observables. 

\begin{algorithm}
    \caption{Projection onto the reachable algebra.}
    \label{algo:reachable}
    \SetAlgoLined
    \Input{A QSO model $(\Bs, \Ys ,\Ac, \Cc, \Sf)$.}
    Compute $\Rs$ using equation \eqref{eq:reachable_space}\;
    Compute $\zentrum\equiv\zentrum(\alg(\Rs))$ and $\Pi_{\zentrum}$\;
    Compute $\bar{\rho}$ according to equation \eqref{eq:rho_bar} and $\sigma = \Pi_\zentrum[\bar{\rho}]$\;
    Compute the reachable algebra $\Ds \equiv \alg\left(\Dc_{\sigma}^{-1}(\Rs)\right)$\;
    Compute the factorizations of $\SE_\Ds$, $\rs$ and $\es$  using equations \eqref{eqn:injection} and \eqref{eqn:reduction}\;
    \Output{$(\check{\Ds}, \Ys, \rs \Ac \es, \Cc\es, \rs\Sf)$}
\end{algorithm}

\begin{theorem}[\bf Reduction on the reachable algebra]
\label{thm:reachable-projection}
    Consider the QSO model $(\Bs, \Ys ,\Ac, \Cc, \Sf)$. Let $\Rs$ be its reachable space and let $\sigma$ be as above.

    Then the {reachable algebra} $\Ds= \alg\left(\Dc_{\sigma}^{-1}(\Rs)\right)$ is the smallest algebra such that there exists $\SE_\Ds$ with ${\rm Im}(\SE_\Ds)=\Dc_\sigma(\Ds)\supseteq\Rs$, and such that for all $\rho_0\in\Sf$ and for all $t\geq 0$ we have,
    \begin{equation}
        \Ac^t[\rho_0] = (\SE_\Ds  \Ac\SE_\Ds)^t\SE_\Ds[\rho_0].
        \label{eqn:reachable_projection}
    \end{equation} 
\end{theorem}
\begin{proof}
    We  divide the proof into three parts: \textbf{i)} In the first part we  prove that the assumption $\SE_\Ds\Pi_\Rs = \Pi_\Rs$ is sufficient for equation \eqref{eqn:reachable_projection} to hold for all $t\geq 0$ and $\rho_0\in\Sf$; \textbf{ii)} In the second part we  prove that the same condition is also necessary; \textbf{iii)} Lastly, we prove that the dynamical algebra is minimal.
    
    \textbf{i)} We  prove the first step by showing that, for all $t\geq0$ and $\rho_0\in\Sf$, both sides of  equation \eqref{eqn:reachable_projection} are equivalent to $(\Pi_\Rs\Ac\Pi_\Rs)^t \Pi_\Rs[\rho_0]$, where $\Pi_\Rs$ is the orthogonal projection onto $\Rs$. Starting from the left side of equation \eqref{eqn:reachable_projection} we can observe, from the definition of $\Rs$, that $\rho_0\in\Rs$ and hence $\rho_0 = \Pi_\Rs[\rho_0]$ for all $\rho_0\in\Rs$. Moreover, from the definition of $\Rs$ we have that $\Rs$ is $\Ac$-invariant, hence $\Ac\Pi_\Rs=\Pi_\Rs\Ac\Pi_\Rs$ and $\Ac^t \Pi_\Rs= (\Pi_\Rs\Ac\Pi_\Rs)^t$. Combining these two observations proves the first equality: \(\Ac^t[\rho_0] = \Ac^t\Pi_\Rs[\rho_0] = (\Pi_\Rs\Ac\Pi_\Rs)^{t}\Pi_\Rs[\rho_0]\)

    Assume now that $\SE_\Ds\Pi_{\Rs} = \Pi_\Rs$. Observing that $\SE_\Ds\Ac\SE_\Ds\Pi_\Rs = \SE_\Ds\Ac\Pi_\Rs = \SE_\Ds\Pi_\Rs\Ac\Pi_\Rs = \Pi_\Rs\Ac\Pi_\Rs$ we obtain $(\SE_\Ds  \Ac\SE_\Ds)^t\Pi_\Rs = (\Pi_\Rs  \Ac{\Pi}_\Rs)^t\Pi_\Rs $. Through similar calculations, we then obtain the second equality \((\SE_\Ds  \Ac\SE_\Ds)^t\SE_\Ds[\rho_0] = (\SE_\Ds  \Ac\SE_\Ds)^t\SE_\Ds\Pi_\Rs[\rho_0] = (\SE_\Ds  \Ac\SE_\Ds)^t\Pi_\Rs[\rho_0] = (\Pi_\Rs  \Ac{\Pi}_\Rs)^t\Pi_\Rs[\rho_0]\).

    \textbf{ii)} Up to this point, we proved that the condition $\SE_\Ds\Pi_\Rs = \Pi_\Rs$ is sufficient for equation \eqref{eqn:reachable_projection} to hold at any time and any initial condition in $\Sf$. However, it is also possible to prove that such a condition is also necessary. Assume that $\SE_\Ds$ is a projector onto a subspace such that equation \eqref{eqn:reachable_projection} holds for all $t\geq0$ and $\rho_0\in\Sf$. First of all, we can observe that requiring that equation \eqref{eqn:reachable_projection} holds for all $t\geq 0$ and $\rho_0\in\Sf$ is equivalent to requiring the same for all $\rho_0\in\Rs$. From the definition of $\Rs$ we have that $\Rs\supseteq\Sf$ and hence it comes trivially that if equation \eqref{eqn:reachable_projection} holds for all $\rho_0\in\Rs$ it holds in particular for $\rho_0\in\Sf$. To prove the opposite implication we can proceed by contradiction. Assume that equation \eqref{eqn:reachable_projection} holds for all $\rho_0\in\Sf$ but there exist $\tau\in\Rs$ and $t\geq0$ such that $\Ac^t[\tau] \neq (\SE_\Ds  \Ac\SE_\Ds)^t\SE_\Ds[\tau]$. Since $\tau\in\Rs$, there exists a set $\{\lambda_{k,\rho_0}\}$ such that $\tau = \sum_k \sum_{\rho_0\in\Sf} \lambda_{k,\rho_0}\Ac^k[\rho_0]$. But then we have $\Ac^t[\tau] = \sum_k \sum_{\rho_0\in\Sf} \lambda_{k,\rho_0}\Ac^{t+k}[\rho_0]$ and, on the other hand 
    \begin{align*}
        &(\SE_\Ds  \Ac\SE_\Ds)^t\SE_\Ds[\tau]=\\ 
        &\qquad= \sum_k \sum_{\rho_0\in\Sf} \lambda_{k,\rho_0} (\SE_\Ds  \Ac\SE_\Ds)^t\SE_\Ds  \underbrace{\Ac^k[\rho_0]}_{(\SE_\Ds  \Ac\SE_\Ds)^k \SE_\Ds[\rho_0]}\\ 
        &\qquad = \sum_k \sum_{\rho_0\in\Sf} \lambda_{k,\rho_0} \underbrace{(\SE_\Ds  \Ac\SE_\Ds)^{t+k} \SE_\Ds[\rho_0]}_{\Ac^{t+k}[\rho_0]}
    \end{align*}
    which leads to an absurd. 
    
    We can then study equation \eqref{eqn:reachable_projection} at $t=0$ for all $\rho_0\in\Rs$, obtaining $(\Ic - \SE_\Ds)[\rho_0] = 0$ for all $\rho_0\in\Rs$. Equivalently we can write $(\Ic - \SE_\Ds)\Pi_\Rs[\rho_0] = (\Pi_\Rs - \SE_\Ds\Pi_\Rs)[\rho_0] =0$ for all $\rho_0\in\Rs$, which is true if and only if $\Pi_\Rs - \SE_\Ds\Pi_\Rs = 0$. This proves that a necessary condition for $\Ac^t[\rho_0] = (\SE_\Ds\Ac\SE_\Ds)^t\SE_\Ds[\rho_0]$ to hold is that $\SE_\Ds\Pi_\Rs = \Pi_\Rs$, or, in other words, $\SE_\Ds$ needs to act as the identity over the reachable space, i.e. $\fix(\SE_\Ds)\supseteq \Rs$. 

    \textbf{iii)} Finally, let us notice that, given $\bar{\rho}$ as defined in equation \eqref{eq:rho_bar} we have that $\supp(\bar{\rho}) = \supp(\Rs)$. Then, from Theorem \ref{thm:minimaldistorted} with $\sigma= \Pi_\zentrum [\bar{\rho}]$, we know that $\Ds_{\sigma} \equiv \Dc_\sigma(\Ds)$, is the smallest distorted algebra that contains $\Rs$ and that allows for a CPTP state extension onto it.  
    Moreover, we have that $\SE_\Ds$ acts as the identity on $\Ds_\sigma$ and, since $\Ds_\sigma\supseteq\Rs$, $\SE_\Ds$ acts as the identity on $\Rs$ as well, hence $\SE_\Ds\Pi_{\Rs} = \Pi_\Rs$. From the previous two points then we have that equation \eqref{eqn:reachable_projection} holds for all $t\geq0$ and $\rho_0\in\Sf$ hence the outputs are also equivalent, i.e. $\Cc\Ac^t[\rho_0] = \Cc\SE_\Ds (\SE_\Ds\Ac\SE_\Ds)^t \SE_\Ds[\rho_0]$ for all $t\geq0$ and $\rho_0\in\Sf$ and thus the reduced model solves Problem \ref{prob:single_time}. 
\end{proof}

To show that Algorithm \ref{algo:reachable} provides a suitable QSO model that solves Problem \ref{prob:single_time} it is then sufficient to combine the results of Theorem \ref{thm:reachable-projection} and Proposition \ref{prop:reduction}.
We thus proved that the reduced model is a valid QSO model and that it satisfies stronger conditions than the one required by Problem \eqref{prob:single_time}, i.e. the reduced model provided by Algorithm \ref{algo:reachable} reproduces the full state dynamics (not only the output one), starting from any state in $\Rs$ (not just $\Sf$). 

On the other hand, notice that the reduced model is not in general guaranteed to reproduce  the correct state dynamics or output starting from {\em any} initial condition $\rho_0\in\Dc_\sigma(\Ds)$. This is because the extension of the reachable state to an algebra might include some non-reachable densities.

\begin{remark}
    We would like to stress that the  dimensional reduction of the matrix representations one obtains from the procedure is twofold: The first comes from the deletion of the off-diagonal blocks in the Wedderburn decomposition \eqref{eq:algblock}, while the second comes from the removal of the repeated blocks (the identity factors). This implies that an effective reduction is achieved also when the reachable algebra $\Ds$ does not present repeated blocks.    
\end{remark}

\subsection{Projection onto the observable algebra}
In this subsection the role previously taken by the reachable algebra is taken by an algebra determined by the observables of interest, which we name \textit{observable algebra} and define as $\Os \equiv \alg(\Ns^\perp)$. We are now ready to present Algorithm \ref{algo:observable}, which also provides a solution to Problem \ref{prob:single_time}. 

\begin{algorithm}
    \caption{Projection onto the observable algebra.}
    \label{algo:observable}
    \SetAlgoLined
    \Input{A QSO model $(\Bs, \Ys ,\Ac, \Cc, \Sf)$.}
    Compute $\Ns$ using equation \eqref{eq:non_observable_subspace}\;
    Compute the observable algebra $\Os=\alg(\Ns^\perp)$\;
    Pick $\sigma$ as any state compatible with $\Os$, e.g. $\sigma = I/n$\;
    Compute the factorizations of $\SE_\Os$, $\rs$ and $\es$  using equations \eqref{eqn:injection} and \eqref{eqn:reduction}\;
    \Output{$(\Check{\Os}, \Ys, \rs\Ac\es, \Cc\es, \rs\Sf)$}
\end{algorithm}

We next prove that Algorithm \ref{algo:observable} not only provides a suitable QSO model that solves Problem \ref{prob:single_time}, but also  that the observable algebra is the smallest algebra that allows the reduced model to reproduce the output dynamics for any initial condition. 

The proof of this theorem is more involved than the proof of Theorem \ref{thm:reachable-projection}, because in this case we have the freedom to choose over the notion of orthogonality with respect to which we construct the complement of $\Ns$. Nonetheless, we are able to show that the natural orthogonality notion provides the optimal solution.

\begin{theorem}[\bf Reduction on the observable algebra]
\label{thm:non_observable-projection}
    Let $(\Bs,\Ys, \Ac, \Cc, \Sf )$ be a QSO model and consider  its non-observable subspace $\Ns.$ The {\em observable algebra} $\Os = \alg(\Ns^\perp)$ is the smallest algebra that allows for a state extension $\SE_\Os$ such that for all $t\geq0$, and for all $\rho_0\in\Bf(\Hc)$ we have 
    \begin{equation}
        \Cc\Ac^t[\rho_0] = \Cc\SE_\Os(\SE_\Os  \Ac\SE_\Os)^t[\rho_0].
        \label{eqn:non_observable_projection}
    \end{equation}  
\end{theorem}
\begin{proof}
     Recall that any choice of modified inner product $\inner{\cdot}{\cdot}_\Sc$ induces an orthogonal  complement of $\Ns$, denoted with $\Ws_\Sc$, i.e. $\Bf(\Hc) = \Ns\bigoplus_\Sc \Ws_\Sc$, and a projection $\Hat{\Pi}_{\Ws_\Sc}$  onto $\Ws_\Sc$, which is  $\Sc$-orthogonal. The latter is such that $\Hat{\Pi}_{\Ws_\Sc}+\Hat{\Pi}_\Ns = \Ic$, the identity super operator, where $\Hat{\Pi}_\Ns$ is the $\Sc$-orthogonal projection onto $\Ns$.

    The proof is divided into three parts: \textbf{i)} In the first part we  prove that a sufficient condition for equation \eqref{eqn:non_observable_projection} to hold is that there exists $\inner{\cdot}{\cdot}_{\Sc}$ such that $\Hat{\Pi}_{\Ws_\Sc}\SE_\Os = \Hat{\Pi}_{\Ws_\Sc}$; \textbf{ii)} In the second part we  prove that the same condition is also necessary; 
    \textbf{iii)} Lastly, we prove that the condition $\Hat{\Pi}_{\Ws_\Sc}\SE_\Os = \Hat{\Pi}_{\Ws_\Sc}$ is actually equivalent to $\fix(\CE_\Os)\supseteq \Ns^\perp$, for all choices of $\inner{\cdot}{\cdot}_\Sc$ and that $\Os$ is minimal. 
    
    \textbf{i)} Let us start by proving that, for all $t\geq0$ and $\rho_0\in\Bf(\Hc)$, for any choice of $\inner{\cdot}{\cdot}_\Sc$, we have $\Cc\Ac^t[\rho_0] = \Cc\Hat{\Pi}_{{\Ws_\Sc}} (\Hat{\Pi}_{{\Ws_\Sc}}\Ac\Hat{\Pi}_{{\Ws_\Sc}})^t \Hat{\Pi}_{{\Ws_\Sc}}[\rho_0]$. From the definition of $\Ns$ we have that $\Ns\subseteq\ker\Cc$, hence $$\Cc = \Cc(\Hat{\Pi}_{\Ws_\Sc} + \Hat{\Pi}_\Ns) = \Cc\Hat{\Pi}_{\Ws_\Sc} + \cancel{\Cc\Hat{\Pi}_\Ns} =\Cc\Hat{\Pi}_{\Ws_\Sc}.$$ Moreover, $\Ns$ is $\Ac$-invariant, i.e. $\Ac\Hat{\Pi}_\Ns = \Hat{\Pi}_\Ns\Ac\Hat{\Pi}_\Ns$, hence, $$\Hat{\Pi}_{\Ws_\Sc}\Ac 
    = \Hat{\Pi}_{\Ws_\Sc} \Ac \Hat{\Pi}_{\Ws_\Sc} + \cancel{\Hat{\Pi}_{\Ws_\Sc} \Hat{\Pi}_\Ns}\Ac \Hat{\Pi}_\Ns 
    = \Hat{\Pi}_{\Ws_\Sc} \Ac \Hat{\Pi}_{\Ws_\Sc} $$ and similarly, by iterating the calculation above $$\Hat{\Pi}_{\Ws_\Sc}\Ac^t 
    = \Hat{\Pi}_{\Ws_\Sc}\Ac\Hat{\Pi}_{\Ws_\Sc} \Ac^{t-1} =\Hat{\Pi}_{\Ws_\Sc} (\Hat{\Pi}_{\Ws_\Sc} \Ac \Hat{\Pi}_{\Ws_\Sc})^t.$$ From these, we obtain: \(\Cc\Ac^t[\rho_0] =  \Cc\Hat{\Pi}_{\Ws_\Sc}\Ac^{t}[\rho_0] = \Cc\Hat{\Pi}_{\Ws_\Sc}(\Hat{\Pi}_{\Ws_\Sc}\Ac\Hat{\Pi}_{\Ws_\Sc})^t[\rho_0]\).
    
    Assume than that there exists $\inner{\cdot}{\cdot}_\Sc$ such that $\Hat{\Pi}_{\Ws_\Sc}\SE_\Os = \Hat{\Pi}_{\Ws_\Sc}$. Then, through similar calculations, we can notice that $\Hat{\Pi}_{\Ws_\Sc}\SE_\Os\Ac\SE_\Os = \Hat{\Pi}_{\Ws_\Sc}\Ac(\Hat{\Pi}_{\Ws_\Sc}+\Hat{\Pi}_\Ns)\SE_\Os  = \Hat{\Pi}_{\Ws_\Sc}\Ac\Hat{\Pi}_{\Ws_\Sc}\SE_\Os + \Hat{\Pi}_{\Ws_\Sc}\Ac\Hat{\Pi}_\Ns\SE_\Os = \Hat{\Pi}_{\Ws_\Sc}\Ac\Hat{\Pi}_{\Ws_\Sc} + \cancel{\Hat{\Pi}_{\Ws_\Sc}\Hat{\Pi}_\Ns}\Ac\Hat{\Pi}_\Ns\SE_\Os = \Hat{\Pi}_{\Ws_\Sc}\Ac\Hat{\Pi}_{\Ws_\Sc}$ and iterating for multiple steps we obtain
    \begin{align*}
         \Hat{\Pi}_{\Ws_\Sc} (\SE_\Os  \Ac\SE_\Os)^{t} 
         &=\Hat{\Pi}_{\Ws_\Sc} \SE_\Os \Ac \SE_\Os (\SE_\Os  \Ac\SE_\Os)^{t-1}\\
         &=\Hat{\Pi}_{\Ws_\Sc} \Ac \Hat{\Pi}_{\Ws_\Sc} (\SE_\Os  \Ac\SE_\Os)^{t-1}\\
         &= \Hat{\Pi}_{\Ws_\Sc}(\Hat{\Pi}_{\Ws_\Sc} \Ac \Hat{\Pi}_{\Ws_\Sc})^t.
     \end{align*}
     It thus follows that:
     \begin{align*}
         \Cc\SE_\Os(\SE_\Os  \Ac\SE_\Os)^t\SE_\Os[\rho_0]&= 
         \Cc \Hat{\Pi}_{\Ws_\Sc}\SE_\Os(\SE_\Os  \Ac\SE_\Os)^t\SE_\Os[\rho_0]\\
         &= \Cc \Hat{\Pi}_{\Ws_\Sc} (\SE_\Os  \Ac\SE_\Os)^{t}\SE_\Os[\rho_0]\\
         &= \Cc \Hat{\Pi}_{\Ws_\Sc}(\Hat{\Pi}_{\Ws_\Sc} \Ac \Hat{\Pi}_{\Ws_\Sc})^t[\rho_0]
     \end{align*}
     for all $t\geq 0$ and $\rho_0\in\Bf(\Hc)$, concluding the first part of the proof.

     \textbf{ii)} We now want to prove that if equation \eqref{eqn:non_observable_projection} holds for all $t\geq0$ and $\rho_0\in\Bf(\Hc)$ then $\Hat{\Pi}_{\Ws_\Sc}\SE_\Os = \Hat{\Pi}_{\Ws_\Sc}$.   We  start by proving that requiring equation \eqref{eqn:non_observable_projection} to hold for all $\rho_0\in\Bf(\Hc)$ and $t\geq0$ is equivalent to require that $\Kc\Ac^t[\rho_0] = \Kc\SE_\Os(\SE_\Os\Ac\SE_\Os)^t[\rho_0]$ holds for all $\rho_0\in\Bf(\Hc)$ and $t\geq0$ and for all output maps $\Kc:\Ac\to\Ys$ such that $\ker\Kc\supseteq\Ns$. The fact that this  implies equation \eqref{eqn:non_observable_projection} comes from the fact that $\Cc$ is one particular output map that satisfies $\ker\Kc\supseteq\Ns$. To prove the opposite, we  proceed by contradiction. Assume that $\Cc\Ac^t[\rho_0] = \Cc\SE_\Os(\SE_\Os  \Ac\SE_\Os)^t[\rho_0]$ holds for all $\rho_0\in\Bf(\Hc)$ and $t\geq0$ but there exist $t\geq0$, $\rho_0\in\Bf(\Hc)$ and $\Kc$ such that $\ker\Kc\supseteq\Ns$ for which $\Kc\Ac^t[\rho_0] \neq \Kc\SE_\Os(\SE_\Os\Ac\SE_\Os)^t[\rho_0]$. Let us recall that the output map $\Cc$ can be written as $\Cc(\cdot) = \sum_i E_i \inner{C_i}{\cdot}_{HS}$ for a set $\{C_i\}\subseteq\As$ and for an orthonormal basis $\{E_i\}\subseteq\Ys$. Similarly, we can write $\Kc(\cdot)=\sum_i E_i \inner{K_i}{\cdot}_{HS}$.  We can then notice:
     \begin{itemize}
        \item $\Ns^\perp = \Span\{\Ac^{\dag t}[C_i], \, t\geq0\}$; 
         \item Necessary and sufficient condition for $\ker\Kc\supseteq\Ns$to hold is that $\Span\{K_i\}\subseteq\Ns^\perp$;
         \item For simplicity, we can assume $\Kc$ to be of the form $\Kc(\cdot) = E\inner{K}{\cdot}_{HS}$ where $E\in\Ys$ and $K\in\Ns^\perp$.
         \item $K\in\Ns^\perp$ implies that there exist a set of coefficients $\{\lambda_{i,k}\}$ such that $K = \sum_{i,k}\lambda_{i,k} \Ac^{\dag t}[C_i]$;
         \item The fact that equation \eqref{eqn:non_observable_projection} holds implies that $\inner{C_i}{\Ac^t[X]}_{HS} = \inner{C_i}{\SE_\Os(\SE_\Os\Ac\SE_\Os)^t[X]}_{HS}$ or, in other words $\inner{\Ac^{\dag t}[C_i]}{X}_{HS} = \inner{(\CE_\Os\Ac^\dag\CE_\Os)^t\CE_\Os[C_i]}{X}_{HS}$, for all $C_i$, for all $t\geq0$ and $X\in\Bf(\Hc)$.
     \end{itemize}
     With these observations, we then have that $\Kc\Ac^t[\rho_0] = E\inner{K}{\Ac^t[\rho_0]} = \sum_{i,k}\lambda_{i,k}E\inner{\Ac^{\dag k}[C_i]}{\Ac^t[\rho_0]} = \sum_{i,k}\lambda_{i,k}E\inner{C_i}{\Ac^{t+k}[\rho_0]}$ and, on the other hand
     \begin{align*}
         &\Kc\SE_\Os(\SE_\Os\Ac\SE_\Os)^t[\rho_0] =\\
         &\qquad= E\inner{K}{\SE_\Os(\SE_\Os\Ac\SE_\Os)^t[\rho_0]}_{HS}\\
         &\qquad= \sum_{i,k}\lambda_{i,k}E\inner{\underbrace{\Ac^{\dag k}[C_i]}_{(\CE_\Os\Ac^\dag\CE_\Os)^k\CE_\Os[C_i]}}{\SE_\Os(\SE_\Os\Ac\SE_\Os)^t[\rho_0]}\\ 
         &\qquad= \sum_{i,k}\lambda_{i,k}E\inner{\underbrace{(\CE_\Os\Ac^\dag\CE_\Os)^{t+k}\CE_\Os[C_i]}_{\Ac^{\dag t+k}[C_i]}}{\rho_0}_{HS}\\
         &\qquad= \sum_{i,k}\lambda_{i,k}E\inner{C_i}{\Ac^{t+k}[\rho_0]}_{HS}.
     \end{align*}
     Comparing the two equations we clearly have that $\Kc\Ac^t[\rho_0]=\Kc\SE_\Os(\SE_\Os\Ac\SE_\Os)^t[\rho_0]$
     which is a contradiction.

    We can thus proceed to prove that if $\Kc\Ac^t[\rho_0] = \Kc\SE_\Os(\SE_\Os\Ac\SE_\Os)^t[\rho_0]$ holds for all $\rho_0\in\Bf(\Hc)$ and $t\geq0$ and for all output maps $\Kc:\Ac\to\Ys$ such that $\ker\Kc\supseteq\Ns$ then it implies that $\Hat{\Pi}_{\Ws_\Sc}\SE_\Os = \Hat{\Pi}_{\Ws_\Sc}$. Since if $\Kc\Ac^t[\rho_0] = \Kc\SE_\Os(\SE_\Os\Ac\SE_\Os)^t[\rho_0]$ holds for any $\Kc$ such that $\ker{\Kc}\supseteq\Ns$ then, in particular it holds for $\Kc=\Hat{\Pi}_{\Ws_\Sc}$, for any choice of $\inner{\cdot}{\cdot}_\Sc$. But then for $t=0$ we have $\Hat{\Pi}_{\Ws_\Sc}[\rho_0] = \Hat{\Pi}_{\Ws_\Sc}\SE_\Os[\rho_0]$ for all $\rho_0\in\Bf(\Hc)$, which holds true if and only if $\Hat{\Pi}_{\Ws_\Sc}\SE_\Os = \Hat{\Pi}_{\Ws_\Sc}$, concluding the second step of the proof 
     
    \textbf{iii)} We  now prove that $\Hat{\Pi}_{\Ws_\Sc}\SE_\Os = \Hat{\Pi}_{\Ws_\Sc}$ is equivalent to $\fix(\CE_\Os)\supseteq \Ns^\perp$ for any choice of $\inner{\cdot}{\cdot}_\Sc$. Applying the adjoint operation on both sides of $\Hat{\Pi}_{\Ws_\Sc}\SE_\Os = \Hat{\Pi}_{\Ws_\Sc}$, we can obtain $\CE_\Os \Hat{\Pi}_{\Ws_\Sc}^\dag = \Hat{\Pi}_{\Ws_\Sc}^\dag$, where we recall that $\SE_\Os^\dag = \CE_\Os$. From the orthogonality of $\Hat{\Pi}_{\Ws_\Sc}$ with respect to $\inner{\cdot}{\cdot}_{\Sc}$, we have that $\Hat{\Pi}_{\Ws_\Sc}^\dag = \Sc\Hat{\Pi}_{\Ws_\Sc}\Sc^{-1}$. Substituting this equality into $\CE_\Os \Hat{\Pi}_{\Ws_\Sc}^\dag = \Hat{\Pi}_{\Ws_\Sc}^\dag$ we obtain $\CE_\Os \Sc\Hat{\Pi}_{\Ws_\Sc}\Sc^{-1} = \Sc\Hat{\Pi}_{\Ws_\Sc}\Sc^{-1}$. Right-applying $\Sc$ and left-applying $\Sc^{-1}$ to both sides of the equation we obtain $\Sc^{-1}\CE_\Os \Sc \Hat{\Pi}_{\Ws_\Sc} = \Hat{\Pi}_{\Ws_\Sc}$, which is equivalent to say that the super-operator $\Sc^{-1}\CE_\Os \Sc$ acts as the identity over ${\Ws_\Sc}$. Observe now that: a) the super-operator $\Sc^{-1}\CE_\Os \Sc$ is a projector; b) as a projector $\Sc^{-1}\CE_\Os \Sc$ acts as the identity over its image, hence its image is equal to the set of its fixed points; c) recalling that $\fix(\CE_\Os)=\Os$ we have $\fix(\Sc^{-1}\CE_\Os \Sc) = \Sc^{-1} \Os$. With this observations we can conclude that $\Sc^{-1}\CE_\Os \Sc$ acts as the identity over ${\Ws_\Sc}$ if and only if ${\Ws_\Sc}$ is contained in $\fix(\Sc^{-1}\CE_\Os \Sc)$, that is  $\Sc^{-1}\Os \supseteq {\Ws_\Sc}$ or, equivalently $\Os\supseteq\Sc{\Ws_\Sc}$. We shall now notice that, for all $\inner{\cdot}{\cdot}_\Sc$, we have ${\Ws_\Sc} = \Sc^{-1}\Ns^\perp$, since $X\in{\Ws_\Sc}$ if $\inner{Y}{X}_\Sc=0$ for all $Y\in\Ns$ and since $\inner{Y}{X}_\Sc=\inner{Y}{\Sc(X)}_{HS}$, by defining $Z = \Sc(X)$ we have that $Z\in\Ns^\perp$ and hence $X\in\Sc^{-1}\Ns^\perp$. Including this relation into $\Os\supseteq\Sc{\Ws_\Sc}$, leads to $\Os\supseteq\Ns^\perp$ or, in other words $\fix(\CE_\Os)\supseteq \Ns^\perp$ . 

     In order for $\SE_\Os$ to be a CPTP projection, $\Os$ needs to be an algebra. In addition, for equation \eqref{eqn:non_observable_projection} to hold, as we showed in {\bf iii)}, $\Os$ must contain $\Ns^\perp.$ Hence,  the smallest algebra that contains $\Ns^\perp$ and for which \eqref{eqn:non_observable_projection} holds is $ \Os = \alg(\Ns^\perp).$ 
\end{proof}

To prove that the reduced model generated by Algorithm \ref{algo:observable} is a solution of our problem is then sufficient to combine the results of Theorem \ref{thm:non_observable-projection} and Proposition \ref{prop:reduction}.  

Notice that any $\SE_\Os$ as in the Theorem statement projects onto a $\Dc_\sigma(\Os),$ parametrized by  a choice of $\sigma\comp\Os$.
Notice also that in the proof of Theorem \ref{thm:non_observable-projection} we show that \eqref{eqn:non_observable_projection} holds for every initial condition $\rho_0\in\Bf(\Hc),$ and not only for the given output map, but also for any $\Kc:\Bs\to\Ys$ such that $\ker\Kc\supseteq\Ns$.
Finally note that $\Os$ need not to be $\Ac^\dag$-invariant for the previous result to hold. In fact, we only need the algebra $\Os$ to contain $\Ns^\perp$, which is itself $\Ac^\dag$-invariant. This implies that if we were to close $\Ns^\perp$ to an $\Ac$-invariant algebra we would obtain a valid QSO model that correctly reproduces the outputs, but it would not necessarily be minimal as such an algebra would certainly contain $\Os$. 

\subsection{Composed reductions}

To solve Problem \ref{prob:single_time}, we apply Algorithms \ref{algo:reachable} and \ref{algo:observable} iteratively, as described in the following Algorithm \ref{algo:composed}. 

\begin{algorithm}
    \caption{Iterative reduction.}
    \label{algo:composed}
    \SetAlgoLined
    \Input{A QSO model $(\Bs, \Ys ,\Ac, \Cc, \Sf)$.}
    Assign $(\As_{0}, \Ys ,\Ac_{0}, \Cc_{0}, \Sf_{0}) = (\Bs, \Ys ,\Ac, \Cc, \Sf)$\;
    Using Algorithm \ref{algo:reachable} on model $(\As_{0}, \Ys ,\Ac_{0}, \Cc_{0}, \Sf_{0})$ compute the model $(\As_1,\Ys, \Ac_1, \Cc_1, \Sf_1)$\;
    Using Algorithm \ref{algo:observable} on model $(\As_{1}, \Ys, \Ac_{1}, \Cc_{1}, \Sf_{1})$ compute the model $(\As_2,\Ys, \Ac_2, \Cc_2, \Sf_2)$\;
    \If{$\dim(\As_0) \neq \dim(\As_2)$}{
    Assign $(\As_{0}, \Ys ,\Ac_{0}, \Cc_{0}, \Sf_{0}) = (\As_2, \Ys ,\Ac_2, \Cc_2, \Sf_2)$\; 
    Go back to step 2\;}
    \Output{$(\As_2, \Ys, \Ac_2, \Cc_2, \Sf_2)$}
\end{algorithm}
Notice that one could obtain an alternative  model-reduction algorithm by applying Algorithm \ref{algo:observable} before Algorithm \ref{algo:reachable} at each iteration. The results of the two reductions are potentially different. In the various numerical tests we ran, we found that applying Algorithm \ref{algo:composed} or one where the order of Algorithm \ref{algo:observable} and Algorithm \ref{algo:reachable} is inverted led to different algebras of the same dimensions.

At each step, Algorithms \ref{algo:reachable} and \ref{algo:observable} provide either a smaller or equivalent model then the previous one. If at some iteration they both do not decrease the dimension, that iteration leaves the model untouched and there is no need to proceed further and we obtain the optimal reduction for our approach. Since we work in a finite-dimensional setting, Algorithm \ref{algo:composed} must converge to its minimum reduction in a finite number of steps.
 
It is worth remarking that, while inspired by the classical algorithm to construct minimal realizations by Rosenbrock \cite{rosenbrock1970state}, the proposed approach needs to expand the reachable and observable subspaces to {\em algebras}: in doing so we need potentially to include non-reachable and non-observable parts of the state space. For this reason, in our setting a single iteration of the algorithm is not sufficient in general to obtain the best reduction. This is demonstrated in the next example.

\begin{example}\label{ex3}
    Consider a QSO model on two qubits i.e. $\Hc\simeq\Cb^{4}$, with $\Sf = \{(I/2+\sigma_x/4)\otimes\tau, (I/2+\sigma_y/4)\otimes\tau\}$ for some full-rank density operator $\tau\in\Cb^{2\times 2}$. Assume that $\Ac=\Ic_{\Cb^{4\times 4}}$, the identity over $\Cb^{4\times 4}$ and that $\Cc(\rho) = \tr(\sigma_z\otimes \sigma_z \rho)$. 

    Computing the reachable space we have $\Rs = \Span\{\Sf\}$ and, picking $\sigma = I\otimes \tau$ we have $\alg(\Dc_\sigma^{-1}(\Rs)) = \Cb^{2 \times 2}\otimes I$. This algebra yields the reduced model $(\Ac',\Cc',\Sf')$, defined on $\Cb^{2\times 2}$ with $\Sf' = \{(I/2+\sigma_x/4), (I/2+\sigma_y/4)\}$ and  $\Ac' = \Ic_{\Cb^{2\times 2}}$ and $\Cc'(\cdot) = \tr(\sigma_z \cdot)\tr(\sigma_z\tau)$. Computing the subspace orthogonal to the non-observable subspace of the reduced model we obtain $\Ns^\perp = \Span\{\sigma_z\}$ hence $\Os = \alg(\Ns^\perp) = \Span\{I, \sigma_z\}$. This algebra then leads to the reduced QSO model $(\Ac'',\Cc'',\Sf'')$ defined on $\Cb^{2\times 2}$ with $\Sf'' = \{I/2\}$, $\Ac'' = \Ic_{\Cb^{2 \times 2}}$ and $\Cc''(\cdot) = \tr(\sigma_z\cdot)\tr(\sigma_z\tau)$. If we then compute the reachable subspace of the model $(\Ac'',\Cc'',\Sf'')$ we obtain $\Rs_2 = \Span\{I\}$ and hence $\alg(\Rs_2) = \Span\{I\}$ which leads to the reduced model $(\check{\Ac},\check{\Cc},\check{\Sf})$ defined on $\Rb$ and with $\check{\Sf} = \{1\}$, $\check{\Ac} = 1$, $\check{\Cc} = \tr(\sigma_z\tau)$. 
    
    This simple example  shows that two steps are not always sufficient to retrieve the minimal model that our method can produce. 
\end{example}

\begin{remark} 
     We shall notice that in Algorithm \ref{algo:composed}, every application of the subroutines, Algorithms \ref{algo:reachable} and \ref{algo:observable}, provides a pair of injection and reduction maps. Let us name them for convenience $\{(\rs_i, \es_i)\}_{i=0,1,\dots}$ and notice that each of those maps is CPTP. Necessarily, both the composition of all the injections, say $\es_* = \es_0 \es_1 \dots$ and of all the reductions $\rs_* = \dots \rs_1 \rs_0$ are CPTP maps, since they are compositions of CPTP maps. Also, $\rs_*:\As\to\check{\As}$ and $\es_*:\check{\As}\to\As$. Then, $\es_*\rs_*$ is a CPTP projection onto a distorted algebra. This shows that there exists a distorted algebra that allows for a CPTP projection (or state extension) onto it, such that the reduction of the model on this distorted algebra provides the same reduced model in a single step as Algorithm \ref{algo:composed} returns in an iterative manner. The question of determining this distorted algebra in a more efficient manner than the one proposed here remains an open problem. 
\end{remark}

\begin{remark}
    Note that, while we were able to prove that Algorithms \ref{algo:reachable} and \ref{algo:observable} provide the minimal algebra that supports the reduced model in the cases where $\Cc=\Ic$ and $\Sf=\Df(\Hc)$ respectively, proving that Algorithm \ref{algo:composed} reaches an algebra of minimal dimension that supports the reduced model remains an open problem. 
    Furthermore one should also note that Proposition \ref{prop:reduction} only provides a sufficient condition that ensures the reduced model is a valid QSO model. In general, it is not necessary that the reduction and injection maps $\Rc$ and $\Jc$ are CPTP. Finding necessary and sufficient conditions for the reduced model to be a valid QSO model also remain an open problem.   
    These problems will be the focus of future work. 
\end{remark}

\section{Examples}\label{sec:examples}
\subsection{Quantum walks: Reduction of Grover's algorithm}
\label{sec:grover}

Grover's algorithm \cite{grover,nielsen_chuang_2010} is a quantum algorithm that solves the unstructured search problem: given a set of $N$ elements and a query function $f(j):\{0,1,\dots,N-1\}\to\{0,1\}$ we want to find an element $j\in\{0,1,\dots,N-1\}$ such that $f(j)=1$. This algorithm is well renowned in the quantum computing community as it provides a quadratic speed-up with respect to its classical counterpart and is also a prototypical example of quantum algorithms built with quantum walks. 

Grover's algorithm can be modeled \cite{cdc2022} as a QSO model defined on $\Bf(\Hc) \simeq \Cb^{N\times N}$: 
\[\begin{cases}
    \rho(t+1) = RO\rho(t)O^\dag R^\dag\\
    \bm p(t) = \diag(\rho(t))
\end{cases}\quad \rho_0 = \ketbra{\psi}{\psi}.\] The initial condition is $\rho_0 = \ketbra{\psi}{\psi}$ where  $\ket{\psi} = \frac{1}{\sqrt{N}}[1,\dots,1]^T\in\Cb^{N}$. The evolution of Grover's algorithm is composed of two unitary evolutions in sequence: first, the oracle $O$ is applied, i.e. a unitary operation such that $O\ket{j} = (-1)^{f(j)}\ket{j}$; second, a reflection $R = 2\ketbra{\psi}{\psi} - I$ is applied. This leads to the discrete-time dynamics $\Ac(\cdot) = RO \cdot O^\dag R^\dag$. The output quantity of interest - the one we want the reduced model to reproduce -  is the population of the state $\rho$, i.e. its diagonal in the standard basis, $\Cc(\cdot) = \diag(\cdot) = \sum_{i=0}^{N-1} \ketbra{i}{i}\cdot\ketbra{i}{i}$. 

To derive the reachable subspace we resort to the approach taken in \cite[Sec. 6.1.3]{nielsen_chuang_2010}. Let us define the set of indexes $S = \{j|f(x)=1\}$ and let $M=|S|$. Define also the two states \[\ket{\alpha} \equiv (N-M)^{-1/2} \sum_{j\notin S} \ket{j},\quad \ket{\beta} \equiv (M)^{-1/2} \sum_{j\in S} \ket{j}\] such that $\ket{\psi} = \alpha_0\ket{\alpha} + \beta_0\ket{\beta}$, with $\alpha_0=\sqrt{(N-M)/N}$ and $\beta_0=\sqrt{M/N}$. Performing the required calculations one can then observe that $U^t\ket{\psi} = \alpha(t)\ket{\alpha} + \beta(t)\ket{\beta}$ with $\alpha(t)=\cos(\frac{2t+1}{2}\theta)$, $\beta(t)=\sin(\frac{2t+1}{2}\theta)$ and $\theta = 2\arccos\left(\sqrt{\frac{N-M}{M}}\right)$. In \cite{cdc2022} it has been shown that assuming $M\neq0,N/2,N$, the first three time instants are sufficient to generate linearly independent operators, and hence $\Rs = \Span\{\ketbra{\alpha}{\alpha}, \ketbra{\beta}{\beta}, \ketbra{\alpha}{\beta}+\ketbra{\beta}{\alpha}\}$. 

One can then observe that $$\alg(\Rs) = \Span\{\ketbra{\alpha}{\alpha}, \ketbra{\beta}{\beta}, \ketbra{\alpha}{\beta}, \ketbra{\beta}{\alpha}\}$$ which does not have full support. In fact, let $U$ be any unitary matrix such that $U^\dag\ket{\alpha} = \ket{0}$ and $U^\dag\ket{\beta} = \ket{1}$, then $\alg(\Rs) = U (\Cb^{2\times 2}\bigoplus \zero_R) U^\dag$ with $\zero_R\in\Cb^{(N-2)\times (N-2)}$. We can then define the non-square isometry $V= [I_2 | \zero_{2\times N-2}] \in\Cb^{2\times N}$ to obtain the reduction and injection maps $\rs(\cdot) = VU^\dag \cdot U V^\dag$ and $\es(\cdot) = UV^\dag \cdot VU^\dag$, which are CPTP over their support. 

These maps allow us to determine the reduced model. Let $\check{\As} = \Cb^{2\times 2}$ and let us denote $\ket{0_2},\ket{1_2}\in\Cb^2$ the standard basis of $\Cb^2$.
The reduced model then takes the form 
\[\begin{cases}
    \check{\rho}(t+1) = \check{U}\rho(t)\check{U}\\
    \bm p(t) = \sum_{i\in S}\ketbra{i}{1_2}\cdot\ketbra{1_2}{i} \sum_{i\notin S} \ketbra{i}{0_2}\cdot\ketbra{0_2}{i}
\end{cases}\] with initial condition 
\[\check{\rho}_0 = \rs(\rho_0) = \frac{1}{N}\begin{bmatrix}
    N-M&\sqrt{(N-M)M}\\\sqrt{(N-M)M}&M
\end{bmatrix}.\]
and unitary matrix \[\check{U} =  \frac{N-2M}{N}I_2 -i\frac{\sqrt{(N-M)M}}{N}\sigma_y.\] Verifying that the reduced model is indeed observable is left to the reader. 

This example shows that Grover's algorithm can be efficiently simulated using a single qubit. Notice however that this does not mean that the unstructured search problem can be solved using a single qubit. This is because in order to compute the change of basis $U$ that reduces the model, one needs to know the set $S$, which is the solution to the unstructured search problem.  

\subsection{Open systems: system-environment with a qubit interface}

Consider a system composed of three finite-dimensional interacting subsystems: the {\em system} $\Hc_S$, the qubit {\em interface} $\Hc_I=\Cb^2$ and the {\em environment}, $\Hc_E$, $\Hc = \Hc_S\otimes\Hc_I\otimes\Hc_E$. 

Consider then the following QSO model:
\[\begin{cases}
    \rho(t+1) = \Ec(U\rho(t)U^\dag)\\
    \tau(t) = \tr_{I,E}(\rho(t)) 
\end{cases}\quad \rho_0\in\Df(\Hc)\]
The CPTP dynamics for our initial QSO model is composed of two parts, a unitary evolution $U\cdot U^\dag$ and a dissipative evolution $\Ec(\cdot)$. The unitary evolution is obtained by integrating over a time $\Delta t$ the Hamiltonian $H = H_S\otimes \sigma_z\otimes I_E + I_S\otimes \sigma_z\otimes H_E$ where $H_S, H_E$ can be any Hamiltonians for the system and  the environment, respectively, thus obtaining $U = e^{-iH\Delta t}$. The dissipative part of the dynamics is a probabilistic bit flip over the qubit interface and a generic probabilistic unitary error over the system and environment, i.e. $\Ec(\cdot) = p_I\sigma_x^{(I)}\cdot\sigma_x^{(I)} + p_{S} U_S \cdot U_S^\dag + p_{E}  U_E\cdot  U_E^\dag + (1-p_I-p_S-p_E) I \cdot I$ where $\sigma_x^{(I)} =  I_{S} \otimes \sigma_x \otimes I_{E}$, while $U_S, U_E$ can be any unitary acting only on the system and environment respectively, and $p_I, p_S, p_E\in[0,1]$, such that $p_S+p_I+p_E\leq1$.  

We are interested in reducing the QSO model we just described for any initial condition $\rho_0\in\Df(\Hc)$. We also consider as output map the partial trace over the interface and the environment: let $\{S_i\}$ be an orthonormal basis for $\Bf(\Hc_S)$, e.g. the generalized Gell-Mann matrices, then $\Cc(\cdot) = \tr_{I,E}(\cdot) = \sum_i S_i \tr(S_i\otimes I_I \otimes I_E \cdot)$. As described in Subsection \ref{sec:non-observable_subspace}, the non-observable subspace is orthogonal to
\(\Ns^\perp = \Span\{\Ac^{\dag t}[S_i\otimes I_{I,E}],\, t\geq0\}.\)

In order to compute $\Ns^\perp$, let us start by noticing that, because of the fact that $[H_S\otimes \sigma_z\otimes I_E, I_S\otimes \sigma_z\otimes H_E]=0$ we can write $U = U_2 U_1$ with $U_1 = e^{-i H_S\otimes \sigma_z\Delta t}\otimes I_E$ and $U_2 = I_S\otimes e^{-i \sigma_z\otimes H_E\Delta t}$. Moreover, since $\sigma_z = \ketbra{0}{0}-\ketbra{1}{1}$ we have $$U_1 = [e^{-i H_S\Delta t}\otimes \ketbra{0}{0} + e^{iH_S\Delta t}\otimes \ketbra{1}{1}]\otimes I_E,$$ $$U_2 = I_S\otimes [\ketbra{0}{0}\otimes e^{-i H_E\Delta t} + \ketbra{1}{1}\otimes e^{i H_E\Delta t}].$$
One can then observe that $\Ec^\dag$ leaves $\Span\{S_i\otimes I_{I,E}\}$ invariant since $\Ec^\dag(S_i\otimes I_{I,E}) = p_S U_S (S_i\otimes I_{I,E}) U_S^\dag + (1-p_S) S_i\otimes I_{I,E}$ and $U_S (S_i\otimes I_{I,E}) U_S^\dag\in\Span\{S_i\otimes I_{I,E}\}$ for any unitary $U_S$ that acts only on the system. Similarly, $U_2^\dag \cdot U_2$ acts as the identity over $\{S_i\otimes I_{I,E}\}$. $U_1^\dag \cdot U_1$, instead, makes the diagonal of the interface qubit, observable at one step. In fact, 
\begin{align*}
    U_1 (S_i\otimes I_{I,E}) U_1^\dag &= [e^{-iH_S\Delta t}S_ie^{iH_S\Delta t}\otimes\ketbra{0}{0} + \\
    &\qquad e^{iH_S\Delta t}S_ie^{-iH_S\Delta t}\otimes\ketbra{1}{1} ]\otimes I_E
\end{align*} 
which is contained in $\Span\{S_i\otimes\ketbra{j}{j}\otimes I_E, \, j=0,1\}.$
In the second step then, one can verify that $\Ec^\dag(\cdot)$, $U_2^\dag \cdot U_2$ and $U_1^\dag\cdot U_1$ leave $\Span\{S_i\otimes\ketbra{j}{j}\otimes I_E, \, j=0,1\}$ invariant and thus  one obtains $\Ns^\perp = \Span\{S_i\otimes\ketbra{j}{j}\otimes I_E,\, i=0,1\}$ which is also an algebra.  Let $W$ be the unitary swap matrix defined by  $W^\dag(A\otimes B)W = B\otimes A$ for all $A\in\mathcal{B}(\mathcal{H}_S)$ and $B\in\mathcal{B}(\mathcal{H}_I)$, then 
$$\Ns^\perp = (W\otimes I_E)[\Bf(\Hc_S)\otimes I_E \bigoplus \Bf(\Hc_S)\otimes I_E](W^\dag\otimes I_E).$$ This leads, after some manipulation, to injection and reduction maps defined as \[\es[\check{X}] = \check{X} \otimes \frac{I_E}{\dim(\Hc_E)}\] for all $\check{X} = X\otimes(\alpha\ketbra{0}{0}+\beta\ketbra{1}{1})$ and $X\in\Bf(\Hc_{S})$ and  
\[\rs(X) = \sum_{j=0,1} \tr_E\left(I_S\otimes\bra{j}\otimes I_E X I_S\otimes \ket{j}\otimes I_E\right)\otimes\ketbra{j}{j}.\]

These maps lead to the reduced model defined over the algebra $\check{\As} = \Span\{S_i\otimes\ketbra{j}{j}\} = W(\Bf(\Hc_S)\bigoplus\Bf(\Hc_S))W^\dag \subseteq \Bf(\Hc_S\otimes \Hc_I)$ and dynamics
\[\begin{cases}
    \check{\rho}(t+1) = \check{\Es}(\check{U}\check{\rho}(t)\check{U}^\dag)\\
    \tau(t) = \tr_I(\check{\rho}(t))
\end{cases}\quad \check{\rho}_0 \in\Df\left(\check{\As}\right)\]
where 
$\check{U} = e^{-i H_S\otimes \sigma_z\Delta t}$ and $\check{\Ec}(\cdot) = p_I\sigma_x^{(I)}\cdot\sigma_x^{(I)} + p_{S} U_S \cdot U_S^\dag + (1-p_I-p_S) I \cdot I$ where $\sigma_x^{(I)} =  I_{S} \otimes \sigma_x$.

This example shows that regardless of the Hamiltonians $H_S$ and $H_E$, the unitaries $U_S$ and $U_E$ and the initial state $\rho_0$, the action of the environment can be completely removed  because it is non-observable from the system state. This is due to the fact that the interface qubit completely decouples the action of the environment with that of the system.  Moreover, the interface qubit behaves classically, since its coherences do not influence in any manner the evolution of the system state. 

\section{Conclusions and outlook}
In this work, we develop a general framework and foundational tools for a theory of CPTP model reduction for quantum dynamics.
While finding minimal linear models can be done by relying on linear system theory, as soon as positivity constraints are imposed the problem becomes challenging, and still open even for classical linear systems \cite{benvenuti}. In quantum engineering, effective reduced models have been proposed for control and filtering purposes, but their properties are hard to characterize \cite{azouit2017towards,rouchon2015efficient}. 
We here propose a systematic way to construct reduced models for quantum dynamics that are guaranteed to be CPTP by construction: the map performing the reduction is essentially a projection onto a distorted algebra, constructed either from the reachable or the observable subspaces, obtained via (CPTP factors of) the dual of a conditional expectation. Alternating the two reductions, reachable and observable, one can exploit the knowledge of both the initial conditions and the output of interest to decrease the size of the description. An open problem remains: is the output of the procedure a CPTP model of minimal dimension for the dynamics of interest? While preliminary numerical tests  suggest this is the case, further work is needed to prove optimality in general.

Other extensions of the presented method, as announced in the introduction, include continuous-time models (described by quantum dynamical semigroups \cite{alicki-lendi}) that have been treated in \cite{grigoletto2025exactmodelreductioncontinuoustime} and where more involved examples are considered, dynamics that include measurement processes (quantum filtering equations and quantum trajectories \cite{bouten2007introduction,benoist2023limit}) that have been treated in \cite{letters,grigoletto2025quantummodelreductioncontinuoustime}, controlled dynamics and, crucially, approximate methods. In fact, for practical applications where the available model is noisy or uncertain, a relaxation of the exact reductions we build here might be more appropriate and lead to smaller models.

\section{Acknowledgements}
T.G. and F.T. wish to thank Lorenza Viola and Augusto Ferrante for motivating and stimulating discussions on the topics of this work.

\appendices
\section{Connection with Model Reduction for HMMs}\label{sec:hmm}
In \cite{tac2023} we proposed an algorithm for reduction of {\em classical} HMMs to construct a distorted algebra  starting from the effective subspace, which provides the reduced model in a single step. This solution can be extended to the case of QSO models with the tools we presented here, but, only works if the projection onto the distorted algebra leaves the subspace $\Ns\cap\Rs$ invariant: otherwise, one is left with a  reduction on just the reachable algebra. This limit  is overcome in this work, in a general non-commutative setting, by introducing an iterative reduction in Algorithm \ref{algo:composed}. The following example showcases this fact.

\begin{example}
Consider a QSO model defined on a four-dimensional abelian algebra $\Bs\subseteq\Cb^{4\times 4}$, $\Bs=\Span\{\ketbra{j}{j}, j=0,1,2,3\}$. Let us consider a trivial dynamics $\Ac(\cdot)=I_4\cdot I_4$. Let then consider the output map \[\Cc(\cdot) = \sum_{j=0,1}\ketbra{j}{j} \bra{\phi_j}\cdot\ket{\phi_j}\] where $\ket{\phi_j} = \ket{j}+\ket{j+2}$. Lastly, let us consider the set of initial conditions 
\begin{align*}
    \Sf = \left\{ I_4/4, 
    \frac{1}{7} \begin{bmatrix} 3&&&\\&0&&\\&&2&\\&&&2\end{bmatrix},
    \frac{1}{20} \begin{bmatrix} 7&&&\\&6&&\\&&3&\\&&&4
    \end{bmatrix}\right\}.
\end{align*}
    
Clearly, $\Ns=\ker\Cc=\Span\{\ketbra{0}{0}-\ketbra{2}{2}, \ketbra{1}{1}-\ketbra{3}{3}\}$ and $\Rs = \Span\{\Sf\}$. With these two subspaces, one can observe that the intersection between the reachable and the non-observable subspace is 
\[\Ns\cap\Rs = \Span\{2\ketbra{0}{0}+\ketbra{1}{1}-2\ketbra{2}{2}-\ketbra{3}{3}\}\] and as an effective subspace one can pick 
\[\Es=\Span\left\{I,\begin{bmatrix}
    5&&&\\&-7&&\\&&1&\\&&&1
\end{bmatrix}\right\}.\]
Since we are in an abelian algebra any positive definite $\sigma\in\Es$ with full support guarantees the existence of a CPTP projection onto the related distorted algebra. In particular, we can choose $\sigma = I_4/4$ thus obtaining $\alg(\Es) = \Span\{\ketbra{0}{0},\ketbra{1}{1},\ketbra{2}{2}+\ketbra{3}{3}\}$ and the state extension 
\[\SE_{\alg(\Es)}(\cdot) = \sum_{j=0,1}\ketbra{j}{j}\cdot\ketbra{j}{j} + \frac{E}{2}\sum_{k=2,3}\bra{k}\cdot\ket{k}\]
where $E=\ketbra{2}{2}+\ketbra{3}{3}$. Unfortunately, in this particular case, we have that $\SE_{\alg(\Es)}(\Rs\cap\Ns)\nsubseteq\Rs\cap\Ns$ and thus reducing the model using the state extension onto $\alg(\Es)$ can not work (does not provide the correct output). For this reason, Algorithm 1 proposed in \cite{tac2023} must resort to the reduction to $\alg(\Rs) = \Bs$ in order to work, which means no reduction is possible.  

On the contrary, the algorithm we proposed here provides a reduced model. Let us start by observing that the first step leads to no reduction since $\alg(\Rs) = \Bs$. In the second step, however, we have that $\Ns^\perp = \Span\{\ketbra{0}{0}+\ketbra{2}{2},\ketbra{1}{1}+\ketbra{3}{3}\}$ which is an algebra. This leads to the reduced model defined over $\check{\As} = \Span\{\ketbra{0_2}{0_2},\ketbra{1_2}{1_2}\} \subset \Cb^{2\times 2}$, where $\ket{0_2},\ket{1_2}\in\Cb^2$ form the standard basis for $\Cb^{2}$ and with trivial dynamics $\check{\Ac}=I_2\cdot I_2$, output map $\check{\Cc}(\cdot) = \sum_{j=0,1} \ketbra{j}{j_2}\cdot\ketbra{j_2}{j}$ and the set of initial conditions reduces to \[\check{\Sf} = \left\{ I_2/2,\frac{1}{7}\begin{bmatrix}5&\\&2\end{bmatrix}, \frac{1}{20}\begin{bmatrix}10&\\&10\end{bmatrix}\right\}.\]

Interestingly enough the dimension of the reduced model obtained in this way is even smaller than the algebra constructed from the effective subspace, as proposed in \cite{tac2023} $\alg(\Es)$, and has in fact the same dimension as the effective subspace. This example thus shows, that, not only the algorithm here proposed works even in conditions where the algorithm proposed in \cite{tac2023} works in a non-optimal way (when $\SE_{\alg(\Es)}(\Rs\cap\Ns)\nsubseteq\Rs\cap\Ns$)   but in certain cases can provide a smaller reduction.  
\end{example}

\section{Instrumental results}
\label{sec:instrumental_results}

The following proposition provides both a way to construct distorted algebras and an intuitive connection between algebras and distorted algebras.  
\begin{proposition}
\label{prop:distorted_algebra}
Let $\mathcal{S}\subset\Bf(\Hc)$ be a set and let $\sigma\in\Hf(\Hc)$ be an Hermitian operator with full support, i.e. $\supp(\sigma)=\Hc$.
Then \[{\alg}_\sigma(\Sc) = \Dc_\sigma(\alg(\Dc_\sigma^{-1}(\mathcal{S}))).\]
Moreover, if $\supp(\Sc) = \Hc$, then $\sigma\in\alg_\sigma(\Sc)$.
\end{proposition}
\begin{proof}
For convenience, let us define, $\As_\sigma\equiv\alg_\sigma(\Sc)$, $\As = \Dc_\sigma^{-1}(\As_\rho)$ and, on the right hand side $\Tilde{\As} \equiv \alg(\Dc_\sigma^{-1}(\Sc))$, $\Tilde{\mathscr{A}}_\sigma \equiv\Dc_{\rho}(\Tilde{\As})$. By definition, we have that $\Tilde{\As}$ is a $*$-algebra, and $\As_\sigma$ is a $\sigma$-distorted $*$-algebra. 

We  start by showing that $\Tilde{\As}_\sigma$ is also a $\sigma$-distorted $*$-algebra and that $\Sc\subseteq\Tilde{\As}_\sigma$. 
Consider then $X,Y\in\Tilde{\As}_\rho$ and $\alpha,\beta\in\Cb$. We have:
\begin{itemize}
    \item $\alpha X+\beta Y\in\Tilde{\As}_\sigma$,  trivially from linearity of operator spaces;
    \item $X^\dag\in\Tilde{\As}_\sigma$, since $X = \sigma^{\um} \bar{X} \sigma^{\um}$ for some $\bar{X}\in\Tilde{\As}$ and thus $X^\dag = \sigma^{\um} \bar{X}^\dag \sigma^{\um}\in\Tilde{\As}_\sigma$, since $\bar{X}^\dag\in\Tilde{\As}$;
    \item $X\cdot_\sigma Y \in\Tilde{\As}_\sigma$ since also $Y = \sigma^{\um} \bar{Y} \sigma^{\um}$ for some $\bar{Y}\in\Tilde{\As}$ and $X\cdot_\sigma Y = \sigma^{\um} \bar{X} \cancel{\sigma^{\um}\sigma^{-1}\sigma^{\um}} \bar{Y} \sigma^{\um} = \sigma^{\um} \bar{X} \bar{Y} \sigma^{\um}\in\Tilde{\As}_\sigma$ since $\bar{X}\bar{Y}\in\Tilde{\As}$;
    \item any sequence of linear combinations, adjoints, and multiplications of elements in $\Tilde{\As}_\sigma$ must also be in $\Tilde{\As}_\sigma$. 
\end{itemize}
This proves that $\Tilde{\As}_\sigma$ is a $\sigma$-distorted $*$-algebra. Moreover, by observing that $\alg(\Sc)\supseteq\Sc$ by definition, we have that $\Tilde{\As}\supseteq \Dc_{\sigma^{-1}}(\Sc)$ and thus $\Tilde{\mathscr{A}}_\sigma \supseteq \Dc_\sigma(\Dc_\sigma^{-1}(\Sc)) = \Sc$. We thus have that $\Tilde{\As}_\sigma$ is a $\sigma$-distorted $*$-algebra that contains $\Sc$. Since, by definition, $\alg_\sigma(\Sc)$ is the smallest $\sigma$-distorted $*$-algebra that contains $\Sc$ we have proven that $\Sc\subseteq\As_\sigma\subseteq\Tilde{\As}_\sigma$. 

We shall notice that $\Dc_\sigma(\cdot)$ is invertible and hence proving that $\As_\sigma = \Tilde{\As}_\sigma$ is equivalent to prove $\As = \Tilde{\As}$. Then, to prove  $\As_\sigma \supseteq \Tilde{\As}_\sigma$ we can show that $\As$ is also a $*$-algebra that contains $\Dc_\sigma^{-1}(\Sc)$ and, since by definition we have that $\Tilde{\As}$ is the smallest $*$-algebra that contains $\Dc_\sigma^{-1}(\Sc)$, this implies that $\Dc_\sigma^{-1}(\Sc)\subseteq\Tilde{\As}\subseteq\As$ and hence $\Sc\subseteq\Tilde{\As}_\sigma\subseteq\As_\sigma$. 

The proof now follows closely the steps of the first part of the proof. Consider $X, Y\in\As$ and $\alpha,\beta\in\Cb$. Then we have:
\begin{itemize}
    \item $\alpha X+\beta Y\in\As$ by the linearity of the operator spaces;
    \item $X^\dag\in\As$ since $X = \sigma^{-\um} \bar{X}\sigma^{-\um}$ for some $\bar{X}\in\As_\sigma$ and thus $X^\dag = \sigma^{-\um} \bar{X}^\dag \sigma^{-\um}\in\As$ since $\bar{X}^\dag\in\As_\sigma$;
    \item $XY\in\As$ since $Y = \sigma^{-\um} \bar{Y} \sigma^{-\um}$ for some $\bar{Y}\in\As_\sigma$ and $XY = \sigma^{-\um} \bar{X} \sigma^{-\um} \sigma^{-\um} \bar{Y} \sigma^{-\um} = \sigma^{-\um} \bar{X} \cdot_\sigma \bar{Y} \sigma^{-\um}\in \As$ since $\bar{X}\cdot_\sigma\bar{Y}\in\As_\sigma$;
    \item any sequence of sums, multiplications and adjoints of elements in $\As$ is in $\As$ as well.
\end{itemize}
This proves that $\As$ is a $*$-algebra. Now, observing that $\alg_\sigma(\Sc)\supseteq\Sc$ by definition, we have that $\As\supseteq\Dc_\sigma^{-1}(\Sc)$ and hence, for the argument above, it holds $\Sc\subseteq\Tilde{\As}_\sigma\subseteq\As_\sigma$ and thus combining the two inclusions, we have $\Sc\subseteq \As_\sigma = \Tilde{\As}_\sigma$ and $\As = \Tilde{\As}$.

To conclude the proof we can observe that $I\in\Tilde{\As}$, (see e.g. \cite[Theorem 2.5]{beny2015algebraic}) and $\Dc_\sigma(I)=\sigma\in \Dc_\sigma(\As) = \Tilde{\As}_\sigma$ and thus $\sigma\in\Tilde{
\As}_\sigma = \As_\sigma$.
\end{proof}

\begin{lemma}
\label{lem:little_blocks}
    Let consider $\Vs\subseteq\Bf(\Hc)$ such that $\alg\Vs=\Bf(\Hc)$. Then $\sigma\comp\alg(\Dc_\sigma^{-1}(\Vs))$ if and only if $\alg(\Dc_\sigma^{-1}(\Vs)) = \Bf(\Hc)$.
\end{lemma}
\begin{proof}
    The ``if'' implication is trivial since any operator $\sigma$ is compatible with $\Bf(\Hc)$. We  now prove the opposite implication  by contradiction.  Let us assume that $\alg(\Dc_\sigma^{-1}(\Vs))\subsetneq\Bf(\Hc)$. Then it must admit a nontrivial decomposition, $\alg(\Dc_\sigma^{-1}(\Vs)) = U (\bigoplus_\ell \Bf(\Hc_{S,\ell})\otimes  I_{m_\ell}) U^\dag$, meaning that either it admits at least two orthogonal components ($\ell$) or one of the identity factors is of dimension greater than one. Thus all $X\in\alg(\Dc_\sigma^{-1}(\Vs))$ can be written as $X = U (\bigoplus_\ell X_{S,\ell}\otimes I_{m_\ell} ) U^\dag$. Moreover, from assumptions we have that  $\sigma\comp\alg(\Dc_\sigma^{-1}(\Vs))$ and thus $\sigma$ has the form $\sigma=U\left( \bigoplus_\ell \sigma_{S,\ell}\otimes \tau_{F,\ell} \right) U^\dag$. Then, since by definition we have $\Vs\subseteq\alg_\sigma\Vs = U\left(\bigoplus_\ell \sigma_{S,\ell}^\um \Bf(\Hc_{S,\ell})\sigma_{S,\ell}^\um\otimes \tau_{F,\ell}\right)U^\dag$, we can say that for the basis elements $\Vs = \Span\{V_i\}$ we can write $V_i = U (\bigoplus_\ell  V_{S,\ell}^i  \otimes I_{m_l}) U^\dag$ with $V_{S,\ell}\in\Bf(\Hc_{S,\ell})$ for all $l$. However, because of the block-diagonal structure of $\{V_i\}$ and the fact that addition, multiplication, and adjoint action leave the block-diagonal structure invariant, we have that $\alg\Vs = U \left(\bigoplus_\ell \alg\{V_{S,\ell}^i\}_i \otimes \alg(\tau_{F,\ell})\right) U^\dag$ where $\alg\{V_{S,\ell}^i\}_i \equiv \alg\{ \sigma_{S,\ell}^\um X_{S,\ell}\sigma_{S,\ell}^\um\}\subseteq \Bf(\Hc_{S,\ell})$. This, however, implies that $\alg\Vs\subsetneq\Bf(\Hc)$, which is a contradiction.
\end{proof}

\begin{lemma}
\label{lem:distorted_algebra_inclusion}
    Consider an operator space $\Vs\subseteq\Bf(\Hc)$ with full support, an positive-definite operator $\mu\in\Bf(\Hc)$ such that $\mu\comp\alg_\mu\Vs$ and consider also another positive-definite operator $\sigma\in\Bf(\Hc)$. If $\Dc_\mu^{-1}(\sigma)\in\alg(\Dc_{\mu}^{-1}(\Vs))$ then $\alg\left(\Dc_\sigma^{-1}(\Vs)\right)\subseteq\Dc_\sigma^{-1}\left(\alg_\mu\Vs\right).$
\end{lemma}
\begin{proof}
    Let us define a basis for $\Vs \equiv \Span\{V_i\}$. Since $\Dc_{\sigma}(\cdot)$ is a linear and invertible map, we have that $\Dc_{\sigma}^{-1}(\Vs) \equiv \Span\{\Hat{V}_i\}$, where $\Hat{V}_i = \sigma^{-\um}{V}_i\sigma^{-\um}$. Similarly, we have that $\Dc_\mu^{-1}(\Vs) = \Span\{\Tilde{V}_i\}$ where $\Tilde{V}_i = \mu^{-\um}V_i\mu^{-\um}$ or, equivalently, $V_i = \mu^{\um}\Tilde{V}_i\mu^{\um}$. Combining the two we obtain
    $\Hat{V}_i = \sigma^{-\um}\mu^\um \Tilde{V}_i \mu^\um\sigma^{-\um}$. By linearity of the operator space we have that the same holds for every element in $\Dc_\sigma^{-1}(\Vs)$, i.e. for every $\Hat{X}\in\Dc_\sigma^{-1}(\Vs)$, there exist $\Tilde{X}\in\Dc_\mu^{-1}(\Vs)$ such that $\Hat{X} = \sigma^{-\um}\mu^\um \Tilde{X} \mu^\um\sigma^{-\um}$. Let us then denote $\Tilde{\sigma} \equiv \Dc_\mu^{-1}(\sigma)$, or equivalently $\sigma = \mu^{\um}\Tilde{\sigma}\mu^{\um}$, from hypothesis we have that $\Tilde{\sigma}\in\alg(\Dc_\mu^{-1}(\Vs))$ and also $\Tilde{\sigma}^{-1}\in\alg(\Dc_\mu^{-1}(\Vs))$. 
    
    We proceed to prove that elements in $\alg\left(\Dc_\sigma^{-1}(\Vs)\right)$ are contained in $\Dc_\sigma^{-1}\left(\alg_\mu\Vs\right) = \Dc_\sigma^{-1}\left(\Dc_\mu(\alg(\Dc_\mu^{-1}(\Vs)))\right)$:
    first notice that $\Hat{X}^{\dag}\in\alg(\Dc_\sigma^{-1}(\Vs))$, where  $\Hat{X}^\dag =\sigma^{-\um}\mu^\um \Tilde{X}^\dag \mu^\um\sigma^{-\um}$ and thus, since $\Tilde{X}^\dag\in\alg(\Dc_\mu^{-1}(\Vs))$, we have that $\Hat{X}^\dag \in \Dc_\sigma^{-1}\left(\Dc_\mu(\alg(\Dc_\mu^{-1}(\Vs)))\right)$.     
    Let then consider a second operator $\Hat{Y}\in\Dc_\sigma^{-1}(\Vs)$, for which $\Hat{Y} = \sigma^{-\um}\mu^\um \Tilde{Y} \mu^\um\sigma^{-\um}$ where $\Tilde{Y}\in\Dc_\mu^{-1}(\Vs)$. We then have that $\Hat{X}+\Hat{Y}\in\alg(\Dc_\sigma^{-1}(\Vs))$ where, by linearity, $\Hat{X}+\Hat{Y} = \sigma^{-\um}\mu^\um (\Tilde{X}+\Tilde{Y}) \mu^\um\sigma^{-\um}$ and thus, since $\Tilde{X}+\Tilde{Y}\in \alg(\Dc_\mu^{-1}(\Vs))$ it holds that $\Hat{X}+\Hat{Y}\in \Dc_\sigma^{-1} \left(\Dc_\mu(\alg(\Dc_\mu^{-1}(\Vs)))\right)$. 
    
    Finally, we can consider the product 
    \begin{align*}
        \Hat{X}\Hat{Y} &= \sigma^{-\um}\mu^\um \Tilde{X} \mu^\um\sigma^{-\um}\sigma^{-\um}\mu^\um \Tilde{Y} \mu^\um\sigma^{-\um}\\ 
        &= \sigma^{-\um}\mu^\um \Tilde{X} \mu^\um \sigma^{-1} \mu^\um \Tilde{Y} \mu^\um\sigma^{-\um}\\
        &= \sigma^{-\um}\mu^\um \Tilde{X} \cancel{\mu^\um \mu^{-\um}}\Tilde{\sigma}^{-1}\cancel{\mu^{-\um} \mu^\um} \Tilde{Y} \mu^\um\sigma^{-\um}\\
        &= \sigma^{-\um}\mu^\um \Tilde{X} \Tilde{\sigma}^{-1} \Tilde{Y} \mu^\um\sigma^{-\um}.
    \end{align*}
    We can thus notice that $\Hat{X}\Hat{Y}\in\alg(\Dc_\sigma^{-1}(\Vs))$ and, since, by assumptions $\Tilde{X},\Tilde{Y}\in\Dc_{\mu}^{-1}(\Vs)$ and $\Tilde{\sigma}^{-1}\in\alg\left(\Dc_{\mu}^{-1}(\Vs)\right)$, we have  $\Tilde{X}\Tilde{\sigma}^{-1}\Tilde{Y}\in \alg(\Dc_\mu^{-1}(\Vs))$ and thus $\Hat{X}\Hat{Y}\in\Dc_\sigma^{-1}\left(\Dc_\mu(\alg(\Dc_\mu^{-1}(\Vs)))\right)$. This concludes the proof.
\end{proof}

\printbibliography{}

\end{document}